\documentclass{lmcs}
\pdfoutput=1
\usepackage[utf8]{inputenc}

\usepackage{lastpage}
\lmcsdoi{22}{2}{5}
\lmcsheading{}{\pageref{LastPage}}{}{}%
{Dec.~20,~2024}{Apr.~16,~2026}{}


\usepackage{hyperref, todonotes}
\theoremstyle{plain}

\usepackage{amssymb, stmaryrd, semantic, listings, array, xspace, cite, xcolor, thmtools}
\newcommand{\highlight}[1]{%
	\colorbox{teal!30}{$\displaystyle#1$}}

\usepackage[export]{adjustbox}
\usepackage{arydshln}
\newcommand{\type}[1]{\langle \mathit{#1}\rangle}
\newcommand{\word}[1]{\text{`\texttt{#1}'}}
\newcommand{\etal}{\emph{et al.}}
\newcommand{\Lname}{AuDaLa\xspace}
\newcommand{\rarr}{\rightarrow}
\newcommand{\Rarr}{\Rightarrow}

\newcommand{\vempty}[0]{\varepsilon}
\DeclareMathSymbol{\sm}{\mathbin}{AMSa}{"39}
\newcommand{\nil}[0]{\mathit{null}}

\newcommand{\Sched}[0]{\mathit{Sc}}
\newcommand{\Structs}[0]{\sigma}

\newcommand{\Env}[0]{\xi}
\newcommand{\Labels}[0]{\mathcal{L}}

\newcommand{\Values}[0]{\mathcal{V}}
\newcommand{\Id}{\mathit{ID}}
\newcommand{\getValue}{\mathit{val}}
\newcommand{\defaultVal}{\mathit{defaultVal}}

\newcommand{\StateSpace}[0]{S_\mathcal{G}}
\newcommand{\Literals}[0]{\mathit{LT}}

\newcommand{\Program}[0]{\mathcal{P}}

\newcommand{\String}[0]{\mathit{String}}

\newcommand{\Expr}[0]{\mathit{E}}

\newcommand{\Nat}[0]{\mathbb{N}}

\newcommand{\Int}[0]{\mathbb{Z}}

\newcommand{\Bool}[0]{\mathbb{B}}

\newcommand{\StructTypes}[0]{\Theta}
\newcommand{\StructType}[0]{\theta}
\newcommand{\SynTypes}[0]{\mathcal{T}}
\newcommand{\SynType}[0]{\mathit{T}}
\newcommand{\interp}[1]{\llbracket #1 \rrbracket}
\newcommand{\true}[0]{\mathit{true}}
\newcommand{\false}[0]{\mathit{false}}
\newcommand{\Stab}[0]{\mathit{s\chi}}
\newcommand{\Stack}[0]{\chi}
\newcommand{\readg}[1]{\textbf{rd}(#1)}

\newcommand{\arr}[0]{\textbf{arr}}
\newcommand{\asize}[0]{\textbf{asize}}
\newcommand{\writev}[1]{\textbf{wr}(#1)}
\newcommand{\readA}[0]{\textbf{rdA}}
\newcommand{\writeA}[0]{\textbf{wrA}}
\newcommand{\cons}[1]{\textbf{cons}(#1)}
\newcommand{\push}[1]{\textbf{push}(#1)}

\newcommand{\Stat}[0]{\mathit{ST}}
\newcommand{\Commands}[0]{\mathcal{C}}
\newcommand{\Operator}[0]{\textbf{op}}
\newcommand{\Ifc}[1]{\textbf{if}(#1)}
\newcommand{\Notc}[0]{\textbf{not}}
\newcommand{\this}[0]{\textbf{this}}
\newcommand{\CList}[0]{\gamma}
\newcommand{\StatList}[0]{S}
\newcommand{\ComList}[0]{\gamma}
\newcommand{\StateP}[0]{\mathit{P}}
\newcommand{\Par}{\mathit{Par}}
\newcommand{\nilL}[1]{\ell^0_{#1}}
\newcommand{\SynOp}[0]{O}

\newcommand{\val}[1]{\mathit{val}(#1)}
\newcommand{\impl}{\mathcal{P}_{\mathit{(T, Z)}}}
\newcommand{\sr}[0]{\Rarr_\Program}

\def\CC{{C\nolinebreak[4]\hspace{-.05em}\raisebox{.4ex}{\tiny\bf ++}}}

\lstdefinelanguage{ADL}{
	keywords = [1]{struct, int, bool, nat, String, if, then, Fix, Iter, else, Array},
	morekeywords=[2]{null, this, true, false},
	sensitive=false, 
	morecomment=[l]{//},
	morecomment=[s]{/*}{*/},
	morestring=[b]"
}

\lstdefinelanguage{gmcommands}{
	morekeywords = {push, rd, wr, cons, if, not, op, this, Fix},
	sensitive=true, 
	morecomment=[l]{//},
	morecomment=[s]{/*}{*/},
	morestring=[b]"
}

\lstdefinestyle{gmcommands}{
	language={gmcommands},
	showstringspaces=false,
	basicstyle=\small,
	keywordstyle=\textbf,
	numberstyle=\tiny\color{codegray},
	stringstyle=\slshape,
	commentstyle=\color{codegray},
	emph={
		val
	},
	emphstyle = \rmfamily\itshape,
	breaklines=true
}

\begin{document}

\title[Expressivity of AuDaLa]{Expressivity of AuDaLa: Turing Completeness and Possible Extensions}
\titlecomment{Extension of `AuDaLa is Turing Complete' as published in the proceedings of FORTE 2024~\cite{franken-audala-2024}}
\author[T.T.P.~Franken]{Tom T.P. Franken\lmcsorcid{0000-0002-1168-5450}}
\author[T.~Neele]{Thomas Neele\lmcsorcid{0000-0001-6117-9129}}
\address{Eindhoven University of Technology, The Netherlands}
\email{\{t.t.p.franken, t.s.neele\}@tue.nl}  

\begin{abstract}
	AuDaLa is a recently introduced programming language that follows the new data autonomous paradigm.
	In this paradigm, small pieces of data execute functions autonomously.
	Considering the paradigm and the design choices of AuDaLa, it is interesting to determine the expressivity of the language. 
	In this paper, we implement Turing machines in AuDaLa and prove that implementation correct. 
	This proves that AuDaLa is Turing complete, giving an initial indication of AuDaLa's expressivity.
	Additionally, we give examples of how to add extensions to AuDaLa to increase its practical expressivity and to better match conventional parallel languages, allowing for a more straightforward and performant implementation of algorithms.
\end{abstract}

\maketitle

\section{Introduction}
Nowadays, performance gains are increasingly obtained through parallelism. 
To make use of this, there are many developments in how to get the hardware to process the program efficiently.
Languages are often designed around that, focusing on threads and processes. 
Recently, \Lname~\cite{franken-autonomous-2023, franken-autonomous-2025} was introduced, which completely abstracts away from threads.
In \Lname, data is \emph{autonomous}, meaning that the data executes its own functions.
It follows the new data autonomous paradigm~\cite{franken-autonomous-2023, franken-autonomous-2025}, which abstracts away from active processor and memory management for parallel programming and instead focuses on the innate parallelism of data. 
This paradigm encourages parallelism by making running code in parallel the default setting, instead of requiring functions to be explicitly called in parallel. 
The paradigm also promotes separation of concerns and a bottom-up design process.

The simplicty, structure and focus on data and understandable programs relates \Lname to domain specific languages and brings up the question to which extent \Lname is generally applicable.
It is therefore interesting to establish how applicable and expressive \Lname is, as \Lname is meant to be general-purpose. 

The applicability of \Lname has been studied in other work by creating a compiler~\cite{leemrijse2023} that can translate \Lname programs to CUDA~\cite{garland-parallel-2008}.
It proved beneficial to base the compiler on a new \emph{weak-memory semantics}~\cite{leemrijse-formalisation-2025} of \Lname.
We furthermore formalised the relation of the weak-memory semantics with the original semantics of~\cite{franken-autonomous-2023}.
The performance of \Lname programs compiled under the weak-memory semantics was sufficient for us to consider \Lname fit for parallel programming in practice.

Then, to establish the expressivity of \Lname, we used the original semantics to establish that \Lname can simulate every Turing machine~\cite{franken-audala-2024}. 
It follows that \Lname is \emph{Turing complete} and can therefore compute all effectively computable functions following the Church-Turing thesis~\cite{copeland-church-turing-1997}. 
This also gives an indication for the theoretical expressivity of the data-autonomous paradigm in general.

The current paper extends our paper from FORTE~2024~\cite{franken-audala-2024} in the following ways:
\begin{itemize}
	\item We provide a self-contained account of \Lname's semantics, including a running example, and recap a few results from~\cite{franken-autonomous-2025} required for our proofs.
	\item We give the full proof for the Turing completeness of \Lname, extending parts of the proof which were only sketched before.
	\item We explore of the \emph{practical expressivity} of \Lname: the extent to which \Lname enables the expression of the user's intentions in the program directly.
	We do this by providing three possible extensions to \Lname, including the necessary additions to the syntax and semantics.
	Two of the extensions deal with extended looping mechanisms and the other implements arrays in \Lname.	
	\item We discuss the adaptations, some of the design choices and other adaptations that are worth considering in the future.
\end{itemize} 

We remark that for all results, we use the original, sequentially consistent semantics~\cite{franken-autonomous-2023}. 
This semantics has been proven equivalent to the weak-memory semantics~\cite{leemrijse-formalisation-2025} when no read-write race conditions occur, which is sufficient for the Turing completeness proof in this paper. 
Though we only define the semantic adaptations of this paper for the original semantics of \Lname, extending them to the weak-memory semantics is reasonably straightforward.


\paragraph{Overview.}
In the current paper, we first give a recap of Turing machines and \Lname in Section~\ref{sec:preliminaries}, followed by the implementation of a Turing machine in \Lname in Section~\ref{sec: TMADL}. We then give some properties of \Lname and some important corollaries for proving programs in \Lname in Section~\ref{sec:prop}, which is based on \cite[Section 7]{franken-autonomous-2025}. These are used to prove that \Lname is Turing complete in Section~\ref{sec: well-formed}.
Then, in Section~\ref{sec: adapt}, we propose three extensions to \Lname.
We conclude in Section~\ref{sec: Con}.

\paragraph{Related Work.}
\Lname is a \emph{data-autonomous} language and related to other data-focused languages, like standard data-parallel languages (CUDA~\cite{garland-parallel-2008} and OpenCL~\cite{chong-sound-2014}), languages which apply local parallel operations on data structures (Halide~\cite{ragan-kelley-halide-2017},\\\textsc{ReLaCS}~\cite{raimbault-relacs-1993}) and actor-based languages (Ly~\cite{ungar-harnessing-2010}, A-NETL~\cite{baba--netl-1995}).

Though the expressivity of actor languages has been studied before~\cite{de-boer-decidability-2012} and there is research into suitable Turing machine-like models for concurrency~\cite{qu-parallel-2017, kozen-parallelism-1976, wiedermann-parallel-1984}, there does not seem to be a large focus on proving Turing completeness of parallel languages.
This can have multiple reasons. 
One of these reasons can be that many of these languages extend other languages, e.g., CUDA and OpenCL are built upon \CC. 
For these languages, Turing completeness is inherited from their base language.
Other languages, for example domain specific languages like Halide~\cite{ragan-kelley-halide-2017}, are simple by design, and if the functionality is more important than the applicability, Turing completeness may be of lesser importance. Some languages are also not Turing complete on purpose~\cite{gibbons-functional-2015, deursen-little-1998}, for example to make automated verification decidable.
As \Lname is not an extension of a sequential language, Turing completeness for \Lname is not automatic. 
It has been noted that Turing completeness also does not establish as much about expressivity as it does in the domain of sequential languages, as Turing completeness does not account for the effects of concurrent and distributive operations, which can lead to differences in behaviour in two Turing complete concurrent languages~\cite{di2009expressiveness}.
Though we agree with this, the fact stands that the full behaviour of a language is a superset of the behaviour of the language in the sequential setting, so establishing \Lname to be Turing complete is a good starting point for analysing the expressivity of \Lname. 
By proving \Lname to be Turing complete, we furthermore establish that even though its simple design and rigid structure may make it seem like \Lname is simple enough to be decidable or domain specific, in reality, \Lname is no less complex than other Turing complete, general languages.

Our proof follows the same line as the proof for the Turing completeness of Circal~\cite{detrey-constructive-2002}.
Other parallel systems that have been proven Turing complete include water systems~\cite{henderson-turing-2021} and asynchronous non-camouflage cellular automata~\cite{yamashita-turing-2017}.

\section{Basic Concepts}
\label{sec:preliminaries}
In this section, we discuss the basics of Turing machines and \Lname.
\subsection{Turing Machines}\label{sec: TM}
We define a Turing machine following the definition of Hopcroft \etal~\cite{hopcroft-introduction-2001}.
Let $\mathbb{D} = \{L, R\}$ be the set of the two directions \textit{left} and \textit{right}. A Turing machine $T$ is a 7-tuple $T = (Q, q_0, F, \Gamma, \Sigma, B, \delta)$, with a finite set of control states $Q$, an initial state $q_0\in Q$, a set of accepting states $F\subseteq Q$, a set of tape symbols $\Gamma$, a finite set of input symbols $\Sigma\subseteq \Gamma$, a blank symbol $B\in \Gamma\setminus\Sigma$ (the initial symbol of all cells not initialized) and a partial transition function $\delta: Q \times \Gamma \nrightarrow Q \times \Gamma \times \mathbb{D}$.

Every Turing machine $T$ operates on an infinite \textit{tape} divided into \textit{cells}. 
Initially, this tape contains an input string $Z = z_0\ldots z_n$ with symbols from $\Sigma$, but is otherwise blank.
The cell the Turing machine operates on is called the \emph{head}. 
We represent the tape as a function $t:\mathbb{Z}\rarr\Gamma$, where cell $i$ contains symbol $t(i)\in\Gamma$. 
In this function, cell $0$ is the head, cells $i$ s.t. $i<0$ are the cells left from the head and cells $i$ s.t. $i>0$ are the cells right from the head.
We restrict ourselves to deterministic Turing machines.	
We also assume the input string is not empty, without loss of generality. 

We define a \emph{configuration} to be a tuple $(q, t)$, with $q$ the current state of the Turing machine and $t$ the current tape function.
Given input string $Z = z_0\ldots z_n$, the \emph{initial configuration} of a Turing machine $T$ is $(q_0, t_Z)$, with $q_0$ as defined for $T$, and $t_Z(i) = z_i$ for $0\leq i \leq n$ and $t_Z(i) = B$ otherwise.

During the execution, a Turing machine $T$ performs \emph{transitions}, defined as: 
\begin{defi}[Turing machine transition]\label{def: conf}
	Let $T = (Q, q_0, F, \Gamma, \Sigma, B, \delta)$ be a Turing machine with input string $Z$ and let $(q, t)$ be a configuration such that $\delta(q, t(0)) = (q', z', D)$, with $D\in \mathbb{D}$. Then $(q, t)\rarr (q', t')$, where $t'$ is defined as
	\[
		t'(i) = 
		\left\{\begin{array}{ll}
			z'&\text{if } i = 1\\
			t(i-1)\;&\text{otherwise}
		\end{array}\right.\text{ if $D{=}L$ and }\\
		t'(i) = 
		\left\{\begin{array}{ll}
			z'&\text{if } i = -1\\
			t(i+1)\;&\text{otherwise}
		\end{array}\right.\text{ if $D{=}R$.}
	\]
\end{defi}

We say a Turing machine $T$ \emph{accepts} a string $Z$ iff, starting from $(q_0, t_Z)$ and taking transitions while possible, $T$ halts in a configuration $(q, t)$ s.t. $q\in F$.
\subsection{\Lname}\label{sec:AuD}
In this section, we give a recap of \Lname concepts, program layout and semantics relevant to this paper. 
As an example throughout the text, we use Listing~\ref{lst: reach}, depicting a program for computing reachability on a small directed graph.
In essence, \Lname is a single instruction-multiple data programming language, which abstracts away from explicit memory management and thread management and puts a large focus on the structure of parallel program design.
By doing this, it keeps the design of the program modular and simple.

An \Lname program contains three parts, which are neatly separated in the code.
Firstly, the definitions of the data types and their parameters are expressed as \emph{structs}. 
In Listing~\ref{lst: reach}, we are computing something on a graph, so we define structs for nodes and edges, where edges have a source and target (given as the parameters \textit{in} and \textit{out}) and nodes can either be reachable or not reachable.
During execution, these structs will be instantiated to \textit{struct instances}, which represent the parallel data elements of the system.
For example, in the reachability example, every node in the graph is be represented by a struct instance of the struct \textit{Node}.
The second part are \textit{steps}, which are defined in context of a struct.
These are functions to be executed in parallel.
The steps are \emph{simple}, in the sense that they cannot contain loops.
For example, in Listing~\ref{lst: reach}, the struct \textit{Edge} contains two steps; one for propagating the reachability property, and the other for initialization.
Lastly, a \emph{schedule} separate from the data system dictates in which order steps are executed, and which loops exist during execution.
The looping mechanism used in \Lname is called a \emph{fixpoint} loop, which repeatedly executes the embedded schedule until the entire system is \emph{stable}.
This is the case when in the last iteration of the loop, no new struct instances have been created and no parameters have been changed.
For example, the schedule of Listing~\ref{lst: reach} first calls the initialization step once, before calling a fixpoint loop over the \textit{reachability} step. 
This loop executes the reachability step until the entire system is stable, which in this case means that no more nodes are set to reachable.
This is the case when nodes one to four are set to reachable, as node five is not reachable.

\begin{lstlisting}[float=t, caption={\Lname code for a reachability program on a small graph}, label={lst: reach}]
struct Node (reach: Bool) {}

struct Edge (in: Node, out: Node) {
	reachability {
		if (in.reach = true) then { 
			out.reach := true;
		}
	}
	init {
		Node node1 := Node(true);
		Node node2 := Node(false);
		Node node3 := Node(false);
		Node node4 := Node(false);
		Node node5 := Node(false);
		
		Edge edge12 := Edge(node1, node2);
		Edge edge13 := Edge(node1, node3);
		Edge edge23 := Edge(node2, node3);
		Edge edge34 := Edge(node3, node4);
		Edge edge51 := Edge(node5, node1);
	} 
}

init < Fix(reachability)\end{lstlisting}

In the schedule, step calls and fixpoints are separated by \emph{barriers} ($<$), s.t. $A<B$ means that subschedule $A$ is executed and only when every struct instance is finished, $B$ is executed. 
Step calls can be global step calls or local step calls. 
For a step $F$, a global step call is denoted by the appearance of $F$ in the schedule, and will cause every struct instance of which the struct has a definition of $F$ to execute $F$ in parallel as defined for their struct. 
Local step calls are denoted as $\StructType.F$ for some struct type $\StructType$, and will cause only the struct instances of $\StructType$ to execute $F$. 

Steps consist of statements, which can be if-then statements, variable assignments (of the form $T\ x := \mathit{val}$ for some new variable $x$ of type $T$) variable updates (of the form $X := \mathit{val}$ for some (possibly referenced) variable $X$) and \emph{constructor statements}, which create new struct instances. 
In the statements, one can use most conventional expressions, like applying binary operators, using brackets and negation, refer to \textbf{null} and $\this$, create struct instances and introduce literals or access variables. 
Variables are accessed using references; variable $x.y$ returns the value of $y$ in the element reached through the reference $x$. 

Henceforth, we only consider \Lname program that are \emph{well-formed}.
Well-formed \Lname programs are well-typed and satisfy the following requirements:
\begin{enumerate}
	\item Identifiers may not be keywords.
	\item A step name is declared at most once within each struct definition.
	\item A parameter name is used at most once within each struct definition.
	\item Names of local variables do not overlap with parameter names of the surrounding struct definition.
	\item Variable assignment statements are only used to declare new local variables.
	\item A local variable is only used after its declaration in a variable assignment statement.
\end{enumerate}
Our example in Listing~\ref{lst: reach} is well-formed, as it is well-typed and fulfills all of the above conditions. A type system that formalises the requirements on well-formed programs can be found in~\cite[Section 4]{franken-autonomous-2025}.

\subsection{\Lname Semantics}
In this section, we explain the semantics of \Lname, using the reachability program of Listing~\ref{lst: reach} as a running example.
In the scope of this section, we denote that program as the program $\mathcal{R}$.
In the semantics, we often use lists of the form $a;\ldots;z$, with empty list $\vempty$.
Here, it is important to note that for any list $L$, $L = L;\vempty$.
We also use function updates: a function $f[x \mapsto y]$ is created from a function $f$, where $f[x\mapsto y](x) = y$ and $f[x\mapsto y](z) = f(z)$ for all $z\neq x$.
We allow the shorthand $f[\{a_1\mapsto b_1, a_2\mapsto b_2,\ldots\}]$ for multiple function updates where all left-hand sides are pairwise distinct, and use the indexed notation $f[a, i\mapsto b]$ if $f(a)$ is a tuple and we want to only change the element at index $i$ of that tuple.
We also assume that $\mathcal{R}$ is parsed and that we have its abstract syntax tree.
In general, this means that polymorphic elements in $\mathcal{R}$, like $42$ or $\nil$ are annotated with their type, as for example $42_\Int$, when $42$ is an integer, or $\nil_T$ for some type $T$.
We divide the explanation of the semantics in three parts: \emph{semantic concepts and sets}, \emph{commands and the state} and \emph{the transition relation}.

\paragraph{Semantic Concepts and Sets.}  
To reason about the syntax, the semantics uses a group of syntactic sets: $\Id$ for variable names and other names from the syntax, $\Literals$ for the set of literals, $\mathit{ST}$ for statements, $\Expr$ for expressions, and $O$, containing all syntactic binary operators. Though we allow all possible identifiers in the set $\Id$, we only use a specific subset. For $\mathcal{R}$, we therefore further refine the set $\Id$ to the sets $\mathit{Var}_\mathcal{R}$, containing variable names used in $\mathcal{R}$, $\Par_\mathcal{R}$, containing parameter names used in $\mathcal{R}$ and $\StructTypes_\mathcal{R}$, which contains all struct type names used in $\mathcal{R}$. Concretely, $\mathit{Var}_\mathcal{R} = \{\text{node1}, \ldots, \text{node5}, \text{edge12}, \ldots, \text{edge51}\}$, $\mathit{Par}_\mathcal{R} = \{\text{reach}, \text{in}, \text{out}\}$ and $\StructTypes_\mathcal{R} = \{\text{Node}, \text{Edge}\}$.

The set of all syntactic types is $\SynTypes = \{\texttt{Nat}, \texttt{Int}, \texttt{Bool}, \texttt{String}\}\cup \Id$, which has an associated set of semantic types $\mathbb{T} = \{\mathbb{N}, \mathbb{Z}, \mathbb{B}, \mathit{String}, \Labels\}$. 
Here, $\Labels$ is the set of \emph{labels}. In the \Lname semantics, at any time during the execution of $\mathcal{R}$ every existing struct instance has a unique identifying \emph{label} which is unique to that struct instance. The set $\Labels$ is assumed adequately large to label all potential struct instances that can exist during an execution of $\mathcal{R}$.

In the semantics of $\mathcal{R}$, we define that every single struct type, in this case $\mathit{Node}$ and $\mathit{Edge}$, has a \emph{null-instance}. This instance represents the default element of the struct type, and has two main purposes. Firstly, the null-instances are by definition \textit{stable}: their parameter values cannot be changed by writes. This makes it easier to design fixpoints, because writes to null-instances do not impact stability and reads from null-instances have a deterministic effect. Secondly, null-instances are able to execute functions just like normal instances, which is useful during initialization. In $\mathcal{R}$, the step \textit{init} will be executed by the null-instance of $\mathit{Edge}$, which exists from the initial state of the program. 

From the syntax of $\mathcal{R}$, it follows that $\Labels$ should contain at least 10 labels, to facilitate the two null-instances and the eight struct instances created by the null-instance of \textit{Edge} when executing \textit{init}. We also define the special subset $\Labels^0\subseteq \Labels$ to contain the \emph{null-labels}, which represent the null-instances.

The set of all values is $\Values = \Labels\sqcup \Int\sqcup \Bool \sqcup \String$, with $\sqcup$ denoting the disjoint union. Note that we consider $\Nat\subset \Int$. For a literal $g$, the semantic value is denoted $\val{g}$. We consider the semantic values of literals straightforwardly defined. For example, the semantic value of $\texttt{5}$ is $\val{\texttt{5}}$, which we assume to be $5$. The semantic value of an occurrence of $\texttt{this}$ in the syntax is defined as the label of the currently executing struct instance, and the semantic value of the null-labels is referred to by $\ell^0_\StructType$ for some struct type $\StructType$. In $\mathcal{R}$, the null-labels are therefore $\ell^0_\mathit{Node}$ and $\ell^0_\mathit{Edge}$. 
We also define a null-value for every syntactic type: $0$ for $\texttt{Nat}$ and $\texttt{Int}$, $\false$ for $\texttt{Bool}$, $\vempty$, the empty string, for $\texttt{String}$ and $\ell^0_T$ for if $T\in\Id$, with $T$ the syntactic type. These null-values can be accessed through a function \textit{DefaultVal}.
	
\paragraph{Commands and the State.} 
During execution, when executing a step $F$, we transform the syntax of $F$ to an intermediate language to make the atomic actions executed in the execution of $\mathcal{R}$ explicit, which allows us to keep the operational semantics of \Lname simple with regards to step execution.
We express these atomic actions as \emph{commands}, which are ordered in the intermediate language in a variant of polish notation, which puts the operators after the values.
These commands are then given to the struct instances to execute.

Formally, a struct instance $s$ with some label $\ell_s$ is a tuple $\langle \StructType, \ComList, \Stack, \Env\rangle$, where $\StructType$ is the struct type, $\ComList$ is a list of commands to be executed, $\Stack$ is a value stack where intermediate values are stored during execution, and $\Env$ is the variable environment.
The commands are given below, with simplified descriptions of the requirements and the actions associated with them (a complete formalisation is given later in this section).
\begin{enumerate}
	\item The push command, $\push{\mathit{v}}$ with value $\mathit{v}$, causes $s$ to push $\mathit{val}$ to its value stack $\Stack$. There is a special variant, $\push{\this}$, for which $s$ pushes $\ell_s$ to $\Stack$.
	\item The read command, $\readg{x}$ with variable $x\in\mathit{Var}_\mathcal{R}\cup\Par_\mathcal{R}$, requires a label $\ell$ on $\Stack$ to denote the location of $v$. The associated action is that $s$ reads the variable $x$ in the variable environment of the struct instance of $\ell$ and puts the resulting value on $\Stack$.
	\item The write command, $\writev{x}$ with variable $x\in\mathit{Var}_\mathcal{R}\cup\Par_\mathcal{R}$, requires a label $\ell$ on $\Stack$ to denote the location of $x$ and a value $v$ on $\Stack$. The associated action is that $s$ writes the value $v$ to the entry for $x$ in the variable environment of the struct instance of $\ell$.
	\item The constructor command, $\cons{\StructType'}$ with some $\StructType'\in\StructTypes_\mathcal{R}$, requires that for every parameter $p_1, \ldots, p_n$ of $\StructType$, there is a value $v_1, \ldots, v_n$ on the stack. The associated action is that $s$ initialises a new struct instance $\langle \StructType, \vempty, \vempty, \Env'\rangle$ s.t. $\Env'(p_i) = v_i$ for all $1\leq i\leq n$ with label $\ell'$ and puts $\ell'$ on $\Stack$.
	\item The if command, $\Ifc{C}$ with a command list $C$, requires a boolean value on $\Stack$. The associated action is that $s$ adds the commands in $C$ to $\ComList$ iff the top value on the value of $\Stack$ is the value $\true$.
	\item The not command, $\Notc$, requires a boolean value on $\Stack$. The associated action is that $s$ negates the top value of the value stack.
	\item The operator command, $\Operator(o)$ with some binary operator $o\in\SynOp$, requires two values compatible with $o$ on $\Stack$. The associated action is that $s$ applies $o$ to the top two values of $\Stack$.
\end{enumerate}

The set of all commands is denoted $\mathcal{C}$.
Note that commands assume certain values on the top of the stack $\Stack$.
These values are guaranteed to be present if a program is well-formed, as shown in~\cite[Theorem 7.14]{franken-autonomous-2025}. If the values are not there, the command cannot execute.

The syntax is transformed into commands through the following function:
\begin{defi}[Interpretation function~\cite{franken-autonomous-2025}]\label{def: interp}
	Let $x, x_1, \ldots, x_n\in \Id$ be variables, $E, E_1, \ldots, E_m\in\Expr$ expressions, $g\in\Literals$ a literal, $\StructType\in\Id$ a struct type, $S \in \Stat$ a statement, $\mathcal{S} \in \Stat^{*}$ a list of statements, $\SynType \in \SynTypes$ a type and $\circ \in \SynOp$ an operator from the syntax. Let the list $x_1;...;x_n$ be the list of $n$ variables from $x_1$~to~$x_n$.
	We define the \emph{interpretation function} $\interp{\cdot}: \Stat^{*}\cup\Expr\rarr\Commands^{*}$ transforming a list of statements or expressions into a list of commands:
    \begin{align*}
        \interp{g}                                       			& = \push{\getValue(g)}                                	\\
        \interp{\mathtt{this}}                          			& = \push{\this}                                       	\\
        \interp{\mathtt{null}_\SynType}                  			& = \push{\defaultVal(\SynType)}                       	\\
        \interp{x_1.\cdots.x_n}                          			& = \push{\this};\readg{x_1};\ldots;\readg{x_n}        	\\
        \interp{!E}                             					& = \interp{E};\Notc                                   	\\
        \interp{E_1 \mathop{\mathtt{op}} E_2}   					& = \interp{E_1};\interp{E_2};\Operator(\circ)   		\\
        \interp{\mathtt{if} \; E \; \mathtt{then} \{\mathcal{S}\}} 	& =  \interp{E};\Ifc{\interp{\mathcal{S}}}				\\
        \interp{T\,x := E}                      					& = \interp{x_1.\cdots.x_0.x := E} 						\\
        \interp{x_1.\cdots.x_n.x := E}          					& = \interp{E};\interp{x_1.\cdots.x_n};\writev{x}       \\ 
        \interp{\StructType (E_1,\ldots,E_m)}   					& = \interp{E_1};\ldots;\interp{E_m};\cons{\StructType} \\  
        \interp{\varepsilon}                  						& = \vempty												\\  
        \interp{S;\mathcal{S}}                                     	& = \interp{S};\interp{\mathcal{S}}                     
    \end{align*}
	Note that if $n = 0$, $\interp{x_1.\cdots.x_n} = \push{\this}$ as a special case. We require $\push{\this}$ to give us the label of the current executing struct instance at the start of dereferencing any pointer.
\end{defi}

As an example, interpreting the syntax of the step $\mathit{reachability}$ in $\mathcal{R}$ leads to the following commands:
\begin{align*}
	&\push{\this};\readg{\text{in}};\readg{\text{reach}};\push{\true};\Operator(=);\Ifc{\\
	&\qquad \push{\true};\push{\this};\readg{\text{out}};\writev{\text{reach}}\\
	&}
\end{align*}

Executions are sequences of transitions between \emph{states}. To reflect the current state of the program, states are a tuple $\langle \Sched, \Structs, \Stab\rangle$. Here, $\Sched$ is the currently still-to-be-executed schedule of the program. The second element $\Structs$ is a partial function from labels to struct instances, the \emph{struct environment}, which contains all the current struct instance tuples and therefore also keeps track of the variables, their values and which commands still have to be executed by which struct instance. This struct environment is what formally links the labels to the struct instance tuples. The last element $\Stab$ of the state is the \emph{stability stack}, which is used to keep track of the currently executing fixpoints and whether they are stable. States are collected in the \emph{state space} $\StateSpace$.		

When starting execution of a fixpoint, a new Boolean value is placed on the stability stack corresponding to this fixpoint. This value is set to true every time the fixpoint iterates, and is set to false whenever a parameter is changed during the execution of the fixpoint. If at the end of a fixpoint iteration the value is still true, the fixpoint terminates and removes the variable from the stability stack.

For example, in $\mathcal{R}$, the first time the schedule reaches $\mathit{Fix}(\mathit{reachability})$, the value $\true$ will be placed on the stability stack. Then $\mathit{reachability}$ is executed. If for some struct instance, $\text{out.reach} := \true$ is executed, the top of the stability stack will be set to $\false$. Then at the end of executing $\mathit{reachability}$, another iteration will be started depending on whether the top of the stability stack is $\false$.
If not, the fixpoint is stable and will terminate.

For every program $\Program$, we define an initial state $\StateP^0_\Program$:
\begin{defi}[Initial state~\cite{franken-autonomous-2025}]\label{def: initial}
	The \emph{initial state} of $\Program$ is $\StateP_{\Program}^0 = \langle\Sched_{\Program}, \Structs_{\Program}^0, \vempty\rangle$, where $\Structs_{\Program}^0(\nilL{\StructType}) = \langle \StructType, \vempty, \vempty, \Env^0_\StructType \rangle$ for all $\StructType \in \StructTypes_{\Program}$ and $\Structs_{\Program}^0(\ell) = \bot$ for all other labels.
\end{defi}
Here, $\Env^0_\StructType$ is defined as $\Env^0_\StructType(p) = \defaultVal(\SynType)$ for all $p \in \Par(\StructType, \Program)$ where $\SynType$ is the type of $p$ and $\Par(\StructType, \Program)$ refers to the parameters of $\StructType$ in $\Program$. For all other variables $x$, $\Env^0_\StructType(x)$ is arbitrary.

In our example, $\Env^0_\mathit{Node}$ yields $\false$ for $\mathit{reach}$ and an arbitrary value for all other inputs, and $\Env^0_\mathit{Edge}$ yields $\ell^0_\mathit{Node}$ for both $\mathit{in}$ and $\mathit{out}$ and an arbitrary value for all other inputs.
The initial state for $\mathcal{R}$ is then $\StateP_{\mathcal{R}}^0 = \langle\mathit{init} < \mathit{Fix(reachability)}, \Structs_{\mathcal{R}}^0, \vempty\rangle$, where $\Structs_{\mathcal{R}}^0(\nilL{\mathit{Node}}) = \langle\mathit{Node}, \vempty, \vempty, \Env^0_\mathit{Node} \rangle$ and $\Structs_{\mathcal{R}}^0(\nilL{\mathit{Edge}}) = \langle\mathit{Edge}, \vempty, \vempty, \Env^0_\mathit{Edge} \rangle$ and $\Structs_{\mathcal{R}}^0(\ell) = \bot$ for all other labels.
		
\paragraph{The Transitions.}
The execution of a program is expressed as a sequence of transitions between states. The \emph{transitions} are expressed as a transition relation $\Rarr_\Program$ containing fifteen transition rules, defined relative to the executing program $\Program$. These come in two categories: transition rules for executing commands and transition rules for dealing with the schedule. They are shown in Figure~\ref{fig: semSystem}. For the transitions of the schedule, we remark that the schedule is interpreted as a list, even though the rules use the syntax of \Lname for schedule notation. Therefore, for example, the \textbf{FixInit} rule is also applicable on a schedule $\mathit{Fix}(\mathit{sc})$, even though it is not followed by a barrier with another schedule behind it.
		
\begin{figure}[h]
	\centering
	\fbox{
		\resizebox{\linewidth}{!}{
{\renewcommand{\arraystretch}{2} 
	\begin{tabular}{c}
		\fbox{
			Command Rules
		} \\
		\begin{minipage}{300px}
			\begin{equation*}
				\inference[(\textbf{ComPush})]{
					\Structs(\ell) = \langle \StructType, \push{\mathit{v}};\CList, \Stack, \Env\rangle
				}{
					\langle \Sched, \Structs, \Stab \rangle\sr \langle \Sched, \Structs[\ell \mapsto \langle \StructType, \CList, \Stack;\mathit{v}, \Env\rangle], \Stab\rangle
				}
			\end{equation*}
		\end{minipage}\hfill
		\begin{minipage}{300px}
			\begin{equation*}
				\inference[(\textbf{ComPushThis})]{
					\Structs(\ell) = \langle \StructType, \push{\this};\CList, \Stack, \Env\rangle
				}{
					\langle \Sched, \Structs, \Stab \rangle\sr\langle \Sched, \Structs[\ell \mapsto \langle \StructType, \CList, \Stack;\ell, \Env\rangle], \Stab\rangle
				}
			\end{equation*}
		\end{minipage} 
		\\ \\
		\begin{minipage}{300px}
			\begin{equation*}
				\inference[(\textbf{ComRd})]{
					\Structs(\ell) = \langle \StructType, \readg{x};\CList, \Stack;\ell', \Env\rangle\\
					\Structs(\ell') = \langle \StructType', \ComList', \Stack', \Env'\rangle
				}{
					\langle \Sched, \Structs, \Stab \rangle\sr\langle \Sched, \Structs[\ell \mapsto \langle \StructType, \CList, \Stack;\Env'(x), \Env\rangle], \Stab\rangle
				}
			\end{equation*}
		\end{minipage}\hfill
		\begin{minipage}{300px}
			\begin{equation*}
				\inference[(\textbf{ComWr})]{
					\Structs(\ell) = \langle \StructType, \writev{x};\CList, \Stack;v;\ell',
					\Env\rangle \\
					\Structs(\ell') = \langle \StructType', \ComList', \Stack', \Env'\rangle \\
					\ell' \notin \Labels^{0} \lor x\notin \Par{}(\StructType', \Program) \\
					\mathit{su} = (x\notin\Par{}(\StructType', \Program)\lor\Env'(x) = v) \\
				}{
					\begin{aligned}
						\langle \Sched, \Structs, \Stab \rangle\sr \langle \Sched, \Structs&[\ell \mapsto \langle \StructType, \CList, \Stack, \Env\rangle][\ell', 4\mapsto \Env'[x\mapsto v]], \\
						&\Stab_1 \land \mathit{su};\dots;\Stab_{|\Stab|} \land \mathit{su}\rangle
					\end{aligned}
				}
			\end{equation*}
		\end{minipage} 
		\\ \\
		\begin{minipage}{300px}
		 	\begin{equation*}
		 		\inference[(\textbf{ComWrNSkip})]{
		 			\Structs(\ell) = \langle \StructType, \writev{x};\CList, \Stack;v;\ell',
		 			\Env\rangle \\
		 			\Structs(\ell') = \langle \StructType', \ComList', \Stack', \Env'\rangle \\
		 			\ell' \in \Labels^{0} \land x\in \Par{}(\StructType', \Program) \\
		 		}{
		 			\begin{aligned}
		 				&\langle \Sched, \Structs, \Stab \rangle\sr \langle \Sched, \Structs[\ell \mapsto \langle \StructType, \CList, \Stack, \Env\rangle], \Stab\rangle
		 			\end{aligned}
		 		}
		 	\end{equation*}
		\end{minipage}\hfill
		\begin{minipage}{300px}
			\begin{equation*}
				\inference[(\textbf{ComNot})]{
					\Structs(\ell) = \langle \StructType, \Notc;\CList, \Stack;b, \Env\rangle
				}{
					\langle \Sched, \Structs, \Stab \rangle\sr\langle \Sched, \Structs[\ell \mapsto \langle \StructType, \CList, \Stack;\neg b, \Env\rangle], \Stab\rangle
				}
			\end{equation*}
		\end{minipage} 
		\\ \\
		\begin{minipage}{300px}
			\begin{equation*}
				\inference[(\textbf{ComOp})]{
					\Structs(\ell) = \langle \StructType, \Operator(\circ);\CList, \Stack;a;b, \Env\rangle
				}{
					\langle \Sched, \Structs, \Stab \rangle\sr\langle \Sched, \Structs[\ell \mapsto \langle \StructType, \CList, \Stack;(a\mathop{o} b), \Env\rangle],\Stab\rangle
				}
			\end{equation*}
		\end{minipage}\hfill
		\begin{minipage}{300px}
			\begin{equation*}
				\inference[(\textbf{ComCons})]{
					\Structs(\ell) = \langle \StructType, \cons{\StructType'};\CList, \Stack;v_1;\ldots;v_n, \Env\rangle \\ \Par{}(\StructType', \Program) = \mathit{p}_1: T_1;...;\mathit{p}_n: T_n \\
					\Structs(\ell')=\bot
				}{
					\begin{aligned}
						\langle \Sched, \Structs, &\Stab \rangle\sr \langle \Sched, \Structs[\{\ell \mapsto \langle sL, \CList, \Stack;\ell', \Env\rangle,
						\\[-4pt]
						&\ell'\mapsto\langle sL', \vempty ,\vempty,\Env^{0}_{\StructType'}[\{\mathit{p}_1\mapsto v_1, \ldots, \mathit{p}_n\mapsto v_n\}]\rangle\}],\\
						&\qquad\false^{|\Stab|}\rangle
					\end{aligned}
				}
			\end{equation*}
		\end{minipage} 
		\\ \\
		\begin{minipage}{300px}
			\begin{equation*}
				\inference[(\textbf{ComIfT})]{
					\Structs(\ell) = \langle \StructType, \Ifc{C};\CList, \Stack;\true, \Env\rangle
				}{
					\langle \Sched, \Structs, \Stab \rangle\sr\langle \Sched, \Structs[\ell \mapsto \langle \StructType, C;\CList, \Stack, \Env\rangle], \Stab\rangle
				}
			\end{equation*}
		\end{minipage}\hfill
		\begin{minipage}{300px}
			\begin{equation*}
				\inference[(\textbf{ComIfF})]{
					\Structs(\ell) = \langle \StructType, \Ifc{C};\CList, \Stack;\false, \Env\rangle
				}{
					\langle \Sched, \Structs, \Stab \rangle\sr\langle \Sched, \Structs[\ell \mapsto \langle \StructType, \CList, \Stack, \Env\rangle], \Stab\rangle
				}
			\end{equation*}
		\end{minipage}
		\\ \\ \hdashline
		\fbox{
			Schedule Rules
		} \\
		\begin{minipage}{330px}
			\begin{equation*}
				\inference[(\textbf{InitG})]{\mathit{Done}(\sigma)}{
					\begin{aligned}
						&\langle F<\mathit{sc}, \Structs, \Stab \rangle\sr\\
						&\qquad \langle \mathit{sc}, \Structs[\{\ell \mapsto \langle \StructType, \interp{\StatList_{\StructType}^{F}}, \vempty, \Env\rangle \mid \Structs(\ell) = \langle \StructType, \varepsilon, \Stack, \Env\rangle\}],\Stab\rangle
					\end{aligned}
				}
			\end{equation*}
		\end{minipage}\hfill
		\begin{minipage}{290px}
			\begin{equation*}
				\inference[(\textbf{InitL})]{\mathit{Done}(\sigma)}{
					\begin{aligned}
						&\langle \StructType.F<\mathit{sc}, \Structs, \Stab \rangle \sr \\
						&\qquad \langle \mathit{sc}, \Structs[\{\ell \mapsto \langle \StructType, \interp{\StatList_{\StructType}^{F}}, \vempty, \Env\rangle\mid \Structs(\ell) = \langle \StructType, \varepsilon, \Stack, \Env\rangle\}],\Stab\rangle
					\end{aligned}
				}
			\end{equation*}
		\end{minipage} 
		\\ \\
		\begin{minipage}{300px}
			\begin{equation*}
				\inference[(\textbf{FixInit})]{\mathit{Done}(\sigma)}{
					\begin{aligned}
						&\langle \mathit{Fix}(\mathit{sc})<\mathit{sc}_1, \Structs, \Stab \rangle\sr\\
						&\qquad \langle \mathit{sc}<\mathit{aFix}(\mathit{sc})<\mathit{sc}_1, \Structs,\Stab;\true\rangle
					\end{aligned}
				}
			\end{equation*}
		\end{minipage}\hfill
		\begin{minipage}{300px}
			\begin{equation*}
				\inference[(\textbf{FixIter})]{\mathit{Done}(\sigma)}{
					\begin{aligned}
						&\langle \mathit{aFix}(\mathit{sc})<\mathit{sc}_1, \Structs,  \Stab;\false \rangle\sr\\
						&\qquad\langle \mathit{sc}<\mathit{aFix}(\mathit{sc})<\mathit{sc}_1, \Structs, \Stab;\true\rangle
					\end{aligned}
				}\\
			\end{equation*}
		\end{minipage} 
		\\ \\ 
		\begin{minipage}{300px}
			\begin{equation*}
				\inference[(\textbf{FixTerm})]{\mathit{Done}(\sigma)}{
					\langle \mathit{aFix}(\mathit{sc})<\mathit{sc}_1, \Structs, \Stab;\true \rangle\sr\langle \mathit{sc}_1, \Structs, \Stab\rangle
				}
			\end{equation*}
		\end{minipage}\hfill
	\end{tabular}
}
		}
	}
	\caption{The semantics of \Lname, ordered by category.}
	\label{fig: semSystem}
\end{figure}

We use (occurrences of) $x$ for variables, $p$ for parameters, $v$ for values and $b$ for Boolean values. The notation $\circ$ refers to a syntactic operator, and we assume that $o$ is a corresponding semantic operator. The notation $C$ refers to a list of commands.  $\Par(\StructType, \Program)$ refers to the parameters of $\StructType$ in the syntax of $\Program$. If $\Structs$ is note defined for $\ell$, this is denoted as $\Structs(\ell) = \bot$. Steps are denoted with $F$ and occurrences of $\mathit{sc}$ are schedules, as is $\Sched$. The symbol $S^F_\StructType$ denotes the list of statements in step $F$ in struct $\StructType$ in $\Program$. The symbol $\mathit{aFix}$ is a special fixpoint notation that refers to a fixpoint encountered before but not yet terminated.

The predicate $\mathit{Done}(\Structs)$ is defined as:
\begin{equation*}
	\mathit{Done}(\sigma) = \forall\ell.(\Structs(\ell) = \bot \vee \exists \StructType, \Stack, \Env. \Structs(\ell) = \langle \StructType, \vempty, \Stack, \Env\rangle)
\end{equation*}
which intuitively means that no existing struct instance still has a command to execute.		

The transition relation for $\mathcal{R}$ is therefore $\Rarr_\mathcal{R}$. Which program the transition system is defined for impacts the rules \textbf{ComWr}, \textbf{ComWrNSkip}, \textbf{InitG} and \textbf{InitL}, as those are dependent on information from the program (parameters and statements in steps respectively). 

With this transition system, we can give the full semantics of any program $\Program$, which is defined as the tuple $\langle \StateSpace, \sr, P_\Program^0\rangle$. The semantics for $\mathcal{R}$ is thus $\langle \StateSpace, \Rarr_\mathcal{R}, P_\mathcal{R}^0\rangle$.

\section{The Implementation of a Turing Machine in \Lname}\label{sec: TMADL}
In this section, we describe the implementation of a Turing machine in \Lname.
Let $T = (Q, q_0, F, \Gamma, \Sigma, B, \delta)$ be a Turing machine and $Z$ an input string.
We implement $T$ and initialize the tape to $Z$ in \Lname.
We assume w.l.o.g. that $Q\subseteq \mathbb{Z}$ with $q_0 = 0$ and that $\Gamma \subseteq \mathbb{Z}$ with $B = 0$.

We model a cell of $T$'s tape by a struct \textit{TapeCell}, with a left cell (parameter \textit{left}), a right cell (\textit{right}) and a cell symbol (\textit{symbol}).
The control of $T$ is modeled by a struct \textit{Control}, which saves a tape head (variable \textit{head}), a state $q\in Q$ (\textit{state}) and whether $q\in F$ (\textit{accepting}).
See Listing~\ref{ex:sched}.
\begin{lstlisting}[float=t,caption={The \Lname Turing machine structure}, label={ex:sched}]
struct (*\textit{TapeCell}*) ((*\textit{left}*): (*\textit{TapeCell}*), (*\textit{right}*): (*\textit{TapeCell}*), (*\textit{symbol}*): Int){} //def. of TapeCell

struct (*\textit{Control}*) ((*\textit{head}*): (*\textit{TapeCell}*), (*\textit{state}*): Int, (*\textit{accepting}*): Bool) {
	(*\textit{transition}*) {
		(*\color{gray} see Listing~\ref{ex:clause} and \ref{ex:transition}*)				//definition of the step "transition"
	}
	(*\textit{init}*) {
		(*\color{gray} see Listing~\ref{ex:init}*)									//definition of the step "init"
	}
}

(*\textit{init}*) < Fix((*\textit{transition}*))       //schedule: run "init" once and then iterate "transition"\end{lstlisting}
The step \textit{transition} in the \textit{Control} struct models the transition function $\delta$.
For every pair $(q, z)\in Q\times \Gamma$ s.t. $\delta(q, z) = (q', z', D)$ with $D\in \mathbb{D}$, \textit{transition} contains a clause as shown in Listing~\ref{ex:clause} (assuming $D = R$). This clause updates the state and symbol, and saves whether the new state is accepting. It also moves the head and creates a new \textit{TapeCell} if there is no next \textit{TapeCell} instance, which we check in line 6. Note that we also check whether the head is not null, for while the null-instance of head cannot update its own parameter, it can still make new \textit{TapeCell}s, which is undesirable. For this, as $z$ can be $\nil$, we need to explicitly check whether \textit{head} is a $\nil$-instance. Note that $B = 0$, and that if $D = L$ the code only minimally changes.
\begin{lstlisting}[float=t, caption={A clause for $\delta(q, z) = (q', z', R)$.}, label={ex:clause}]
if ((*\textit{state}*) == (*$q$*) && (*\textit{head.symbol}*) == (*$z$*)) then {
	(*\textit{head.symbol}*) := $z'$; 				 //update the head symbol
	(*\textit{state}*) := $q'$; 								  //update the state
	(*\textit{accepting}*) := $(q'\in F)$; 		//the new state is accepting or rejecting
	
	if ((*\textit{head}*) != null && (*\textit{head.right}*) == null) then {
		(*\textit{head.right}*) := (*\textit{TapeCell}*)((*\textit{head}*), null, 0); //call constructor to create a new TapeCell
	}
	(*\textit{head}*) := (*\textit{head.right}*);         //move right
}\end{lstlisting}
	
The clauses for the transitions are combined using an if-else if structure (syntactic sugar for a combination of ifs and variables), so only one clause is executed each time \textit{transition} is executed.
See Listing~\ref{ex:transition}.
\begin{lstlisting}[float = t, caption={The \textit{transition} step. The shown pairs all have an output in $\delta$.}, label={ex:transition}]
(*\textit{transition}*) {
	if ((*\textit{state}*) == (*$q_1$*) && (*\textit{head.symbol}*) == (*$z_1$*)) then{
		/*clause 1*/ 
	}
	else if ((*\textit{state}*) == (*$q_2$*) && (*\textit{head.symbol}*) == (*$z_2$*)) then {
		/*clause 2*/ 
	}
	else if ((*\textit{state}*) == (*$q_3$*) && (*\textit{head.symbol}*) == (*$z_3$*)) then {
		/*clause 3*/ 
	}
	// etc.
} 	\end{lstlisting}
	\begin{lstlisting}[float=t, caption={Initializing input string $Z$.}, label={ex:init}]
(*\textit{init}*) {
	(*\textit{TapeCell}*) cell0 := (*\textit{TapeCell}*)(null, null, $z_0$);            // initialize the tape
	(*\textit{TapeCell}*) cell1 := (*\textit{TapeCell}*)(null, null, $z_1$);
	(*\textit{TapeCell}*) cell2 := (*\textit{TapeCell}*)(null, null, $z_2$);
	
	cell1.(*\textit{left}*) := cell0;                                            // connect the tape
	cell0.(*\textit{right}*) := cell1;
	cell2.(*\textit{left}*) := cell1;
	cell1.(*\textit{right}*) := cell2;
	
	(*\textit{Control}*)(cell0, 0, $(q_0\in F)$);                               //initialize the control
}\end{lstlisting}
In the step \textit{init} in the \textit{Control} struct, we create a \textit{TapeCell} for every symbol $z\in Z$ from left to right, which are linked together to create the tape.
We also create a \textit{Control}-instance.
Listing~\ref{ex:init} shows this for an example tape $Z = z_0,z_1,z_2$.
A call of \textit{init} in the schedule causes the $\nil$-instance of \textit{Control} to initialize the tape.
It also initializes a single non-$\nil$ instance of \textit{Control}.
The schedule will then make that instance of \textit{Control} run the \textit{transition} step until the program stabilizes.
Listing~\ref{ex:sched} shows the final structure of the program.
	
\begin{exa}\label{ex:tm}
As an example, consider a Turing machine
$T_\mathit{ex} = (\{0, 1\}, 0, \{1\}, \{0, 1, 2\}, \{1, 2\},\linebreak 0, \delta)$, with $\delta(0, 0) = (0, 0, R)$, $\delta(0, 1) = (0, 1, R)$ and $\delta(0, 2) = (1, 2, R)$, with no transitions defined for state $1$. This Turing machine walks through the input and stops after finding a $2$. In Listing~\ref{lst: tmex}, we have created an \Lname program for this Turing machine with input string `$1121$'. The Turing machine with input string `$1121$' is visualised in Figure~\ref{fig:tm}.
\end{exa}
\begin{figure}[t]
	\includegraphics[width=0.75\textwidth]{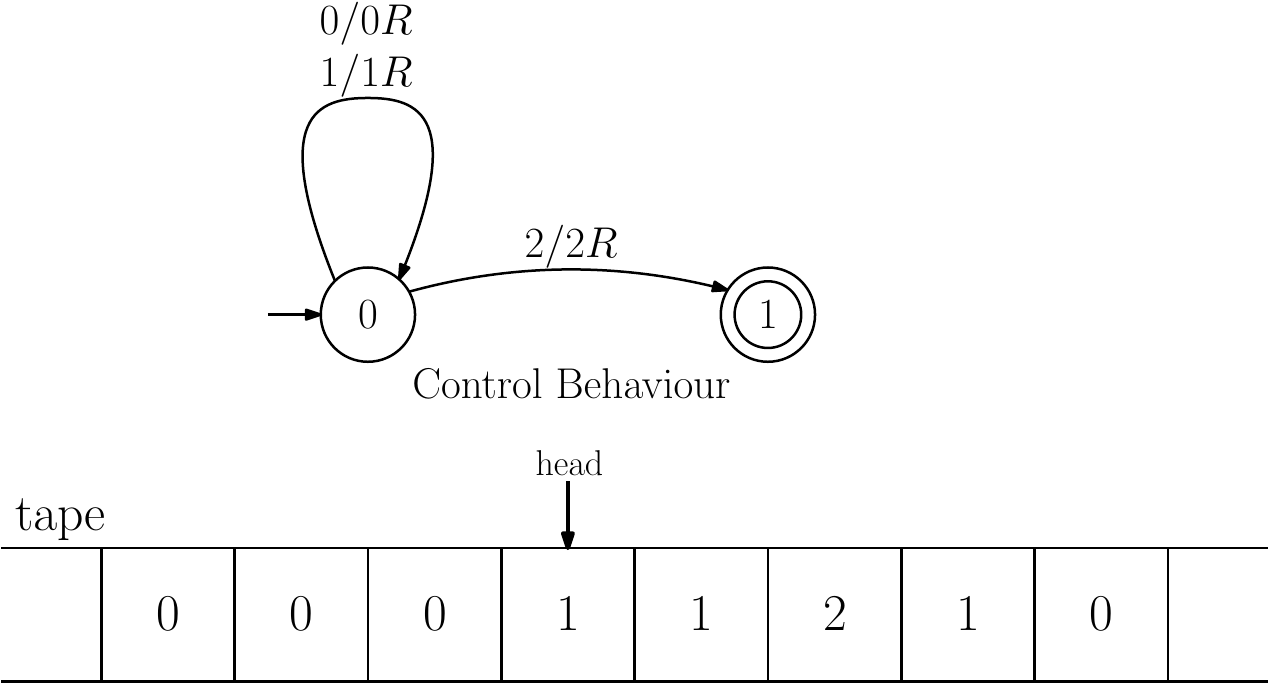}
	\caption{The Turing Machine of Example~\ref{ex:tm} with input string `$1121$'.}
	\label{fig:tm}
\end{figure}

\begin{lstlisting}[float=p, caption={The \Lname program for $T_\mathit{ex}$ with input string $1121$.}, label={lst: tmex}]
struct (*\textit{TapeCell}*) ((*\textit{left}*): (*\textit{TapeCell}*), (*\textit{right}*): (*\textit{TapeCell}*), (*\textit{symbol}*): Int){} //def. of TapeCell
struct (*\textit{Control}*) ((*\textit{head}*): (*\textit{TapeCell}*), (*\textit{state}*): Int, (*\textit{accepting}*): Bool) {
	(*\textit{transition}*) {
		if ((*\textit{state}*) == 0 && (*\textit{head.symbol}*) == 0) then{ // transition (*{\color{gray} $\delta$}*)(0,0)
			(*\textit{head.symbol}*) := 0; 				
			(*\textit{state}*) := 0; 				
			(*\textit{accepting}*) := false; 		
			if ((*\textit{head}*) != null && (*\textit{head.right}*) == null) then {
				(*\textit{head.right}*) := (*\textit{TapeCell}*)((*\textit{head}*), null, 0); 
			}
			(*\textit{head}*) := (*\textit{head.right}*);         
		}
		else if ((*\textit{state}*) == 0 && (*\textit{head.symbol}*) == 1) then{ // transition (*{\color{gray} $\delta$}*)(0,1)
			(*\textit{head.symbol}*) := 1; 				
			(*\textit{state}*) := 0; 								  
			(*\textit{accepting}*) := false; 		
			if ((*\textit{head}*) != null && (*\textit{head.right}*) == null) then {
				(*\textit{head.right}*) := (*\textit{TapeCell}*)((*\textit{head}*), null, 0); 
			}
			(*\textit{head}*) := (*\textit{head.right}*);         
		}
		else if ((*\textit{state}*) == 0 && (*\textit{head.symbol}*) == 2) then{ // transition (*{\color{gray} $\delta$}*)(0,2)
			(*\textit{head.symbol}*) := 2; 				
			(*\textit{state}*) := 1; 								  
			(*\textit{accepting}*) := true; 		
			if ((*\textit{head}*) != null && (*\textit{head.right}*) == null) then {
				(*\textit{head.right}*) := (*\textit{TapeCell}*)((*\textit{head}*), null, 0);
			}
			(*\textit{head}*) := (*\textit{head.right}*);        
		}
	}
	(*\textit{init}*) {
		(*\textit{TapeCell}*) cell0 := (*\textit{TapeCell}*)(null, null, $1$);            // initialize the tape
		(*\textit{TapeCell}*) cell1 := (*\textit{TapeCell}*)(null, null, $1$);
		(*\textit{TapeCell}*) cell2 := (*\textit{TapeCell}*)(null, null, $2$);
		(*\textit{TapeCell}*) cell3 := (*\textit{TapeCell}*)(null, null, $1$);
		cell1.(*\textit{left}*) := cell0;                                            // connect the tape
		cell0.(*\textit{right}*) := cell1;
		cell2.(*\textit{left}*) := cell1;
		cell1.(*\textit{right}*) := cell2;
		cell3.(*\textit{left}*) := cell2;
		cell2.(*\textit{right}*) := cell3;
		(*\textit{Control}*)(cell0, 0, false);        
	}
}
(*\textit{init}*) < Fix((*\textit{transition}*))       //schedule: run "init" once and then iterate ``transition"\end{lstlisting}
	
\section{Properties of \Lname Programs}\label{sec:prop}
This section discusses some properties of \Lname programs that are relevant for the Turing completeness proof. This section is a partial recap of Section 7 of~\cite{franken-autonomous-2025}. In Section~\ref{sec: standard}, we recap some standard properties of \Lname programs. In Section~\ref{sec: well-formed} we summarize those results of~\cite[Theorem 7.14]{franken-autonomous-2025} which are relevant to this paper. This theorem is not included in full because it is out of the scope of this paper.

\subsection{Standard Properties}\label{sec: standard}
First, as \Lname uses a reference notation of the form $x.y$ to access variable $y$ from the struct instance saved under variable $x$, we define an \emph{semantic reference notation}, so that we can use the notation freely in both the syntactic and semantic contexts.

\begin{defi}[Reference Notation]
	Let $\Structs$ be a structure environment and let $x$ be a struct instance with label $\ell_x$ such that $\Structs(\ell_x) = \langle \StructType, \ComList, \Stack, \Env\rangle$. We define the notation $\ell_x.a^i$ inductively on $i$, with a variable $a$:
	\begin{enumerate}[nolistsep]
		\item $\ell_x.a^1 = \xi(a)$,
		\item $\ell_x.a^i = \Env'(a)$ for $i>1$, with $\ell_x.a^{i-1}\in\Labels$ and $\Structs(\ell_x.a^{i-1}) = \langle \StructType', \ComList', \Stack', \Env'\rangle$.
	\end{enumerate}
	We write $\ell_x.a$ for $\ell_x.a^1$ and $x.a$ for $\ell_x.a$, where $\Structs(\ell_x) = x$ and $a$ is a parameter of $x$.
\end{defi}

To give the following two definitions, we first define that an \emph{execution} of a step $F$ from $P$ is a chain of transitions starting with an \textbf{Init} transition for $F$ from $P$ and ending at the first state with struct environment $\sigma$ s.t. $\mathit{Done}(\sigma)$ holds. We say that $P$ can execute $F$ if such a chain exists from $P$. Additionally, note that while the \Lname semantics does not define which labels get assigned to new struct instances beyond that they should be new, the exact labels which get assigned to new struct instances do not have an impact on the program: they aren't null-labels, and labels which are not null-labels are only used abstractly. We therefore define determinism to not take newly assigned labels into account. The definition of determinism for \Lname is then:
\begin{defi}[Determinism]
	Let $F$ be a \textit{step} in an \Lname program.
	Then $F$ is \emph{deterministic} for some state $P$ iff $P$ can execute $F$ and there exists exactly one state that is reached by executing $F$ modulo the labels newly assigned to struct instances during $F$. 
\end{defi}

We follow it up with a definition for \emph{race conditions} in \Lname. We consider an \emph{access} of a variable $x$ to be an application of \textbf{ComRd} or \textbf{ComWr} with $x$ as the parameter of the command that induced that transition.		
\begin{defi}[Race Conditions]
	Let $F$ be a step of some \Lname program $\Program$. Let $P$ be a state of $\Program$. Then $F$ contains a \emph{race condition} starting in $P$ iff there is an enabled \textbf{InitL} or \textbf{InitG} transition for $F$ from $P$ and during the execution of $F$ after taking this \textbf{Init} transition, there exist a parameter $x$ which is accessed by two distinct struct instances $a$ and $b$, with at least one of these accesses writing to $x$. We call a race condition between writes a \emph{write-write} race condition, and a race condition between a read and a write a \emph{read-write} race condition.
\end{defi}

We combine these two concept in the following lemma.

\begin{lem}[\Lname Determinism~\cite{franken-autonomous-2025}]\label{lem:det}
	An \Lname step $F$ is \emph{deterministic} for some state $P$ if $F$ does not contain a race condition starting in $P$.
\end{lem}
It follows that when a step is deterministic we can ignore interleaving of struct instances during the execution of the step when determining the effects of the step.

\subsection{Properties of Well-formed \Lname Programs}\label{sec: well-formed}		
When programs are well-formed, we know even more about them. In this section, give some results derived from Theorem 7.14 in~\cite{franken-autonomous-2025}, which are important for the Turing completeness proof. Theorem 7.14 is omitted, as stating it and its proof fully is expansive and has no purpose in this paper.

We define that the state after the execution of an expression $E$ or a statement $S$ by some struct instance $s$ is the state resulting from the transition of the last command from $\interp{E}$ or $\interp{S}$ as put in the command list of $s$. Then we know that for any well-formed program $\Program$, the following properties hold:
\begin{cor}[Progress]\label{cor:prog}
	The execution of a well-formed program never gets stuck unless it terminates: every sequence of transitions from a state $P_1$ with struct environment $\Structs_1$ s.t. $\mathit{Done}(\Structs_1) = \false$ eventually reaches a state $P_1'$ with struct environment $\Structs_1'$ s.t. $\mathit{Done}(\Structs_1') = \true$, and every sequence of transitions from a state $P_2$ with struct environment $\Structs_2$ s.t. $\mathit{Done}(\Structs_2) = \true$ eventually reaches a state $P_2'$ with struct environment $\Structs_2$ where either the schedule is empty or $\mathit{Done}(\Structs_2') = \false$.
\end{cor}
\begin{cor}[Expression Results]
	Every expression $E$, when executed by $s$ during a step $F$ in a well-formed program $\Program$, results in a single value $v$ on the stack of $s$. This value is of the semantic type expected from the syntax of $E$ and it is deterministic if there are no read-write race conditions.
\end{cor}
\begin{cor}[Execution Effects]\label{cor: statef}
	Let $S$ be a statement executed during the execution of a step $F$ in a well-formed program $\Program$ by a struct instance $s$. Let $P'$ be the state resulting from the transition induced by the last command of $\interp{S}$ for $s$. Then:
	\begin{enumerate}
		\item If $S$ is an update statement $x_1.\ldots.x_n.x := E$, then, with $x$ the variable in the struct instance belonging to a result of an execution of $x_1.\ldots.x_n$, if $x$ is not a parameter of a null-instance, $x$ will be updated to a result of an execution of $E$ in $P'$, and all values in the stability stack will be reset to $\false$ if this means that a parameter has changed. If $x$ is a parameter of a null-instance, $E$ is executed and its effects are still present in $P'$, but $x$ has not been updated.
		\item If $S$ is a variable assignment statement $T\ x:= E$, then it has the same effect as executing the update statement $x:= E$.
		\item If $S$ is an if-statement, in $P'$, $E$ will be executed to a result $b$. If $b = \true$, the if-clause will be taken. If $b = \false$, the if-clause will be skipped.
		\item If $S$ is a constructor statement $\StructType(E_1, \ldots, E_m)$, then in $P'$, there will be a new struct instance for a fresh label $\ell$ with type $\StructType$ and its parameters set to results of executions of $E_1, \ldots, E_m$, and all values in the stability stack will be reset to $\false$. Additionally, $\ell$ will be the last value on the stack of $s$.
	\end{enumerate}
	Additionally, let $\mathcal{E}$ be the expressions executed as a part of $S$. For any expression $E\in\mathcal{E}$, if $E$ is a constructor expressions $\StructType(E_1, \ldots, E_m)$, its result will be a fresh label $\ell'$ s.t. in $P'$, there is a struct instance for $\ell'$ with type $\StructType$ and its parameters set to results of executions of $E_1, \ldots, E_m$. Furthermore, if there exists a constructor expression $E$ in $\mathcal{E}$, all values in the stability stack will be reset to $\false$. 
	Lastly, the effects of a statement will be deterministic if all results from all expressions in $S$ are deterministic.
\end{cor}
These corollaries and results allow us to simplify the proof of \Lname programs by keeping the proof on the level of the syntax.

\section{Turing Completeness}\label{sec: exe}
In this section, we show that \Lname is Turing complete. To do this, we establish an equivalence between the results of transitions taken by $T$ with input string $Z$ and the execution the \textit{transition} step of its implementation.
Henceforth, let $\impl$ be the implementation of a Turing machine $T$ with an input string $Z= z_0\ldots z_n$ as specified in Section~\ref{sec: TMADL}. As $\impl$ is well-typed and satisfies the requirements given in Section~\ref{sec:AuD}, $\impl$ is well-formed.

As $\impl$ is well-formed, it does not get stuck, according to Corollary~\ref{cor:prog}. It follows that the execution of $\impl$ consists of the execution of a step \textit{init} from the initial state, and then executions of the step \textit{transition}. Between these executions, there will be one or more state with a struct environment $\Structs$ s.t. $\mathit{Done}(\Structs)$ holds. These states we call \emph{idle states}. For these idle states, we define \emph{implementation configurations}:
\begin{defi}[Implementation Configuration]\label{def: implconf}
Let $P$ be an idle state of $\impl$ containing a single non-$\nil$ instance $c$ of \textit{Control}.
Then we define the \emph{implementation configuration} of $P$ as a tuple $(q_P, t_P)$ s.t. $q_P$ is the value of the \textit{state} parameter of $c$ and $t_P: \mathbb{Z}\rarr \mathbb{Z}$ defined as:
\[
	t_P(i) = \begin{cases}
		c.\mathit{head}.\mathit{symbol} &\text{if } i = 0\\
		c.\mathit{head}.\mathit{left}^{\text{-}i}.\mathit{symbol} &\text{if } i < 0\\
		c.\mathit{head}.\mathit{right}^i.\mathit{symbol} &\text{if } i > 0
	\end{cases},
\]
\end{defi}

For our equivalence, we compare the Turing machine configuration of $T$ with input string $Z$ after doing transitions to the implementation configuration of $\impl$. As the implementation configuration is defined on the struct environment of idle states only, we will abstract from which idle state we take between executions of the \textit{transition} step, as they will all have the same configuration: there is no semantic transition that changes the struct environment enabled when the $\mathit{Done}$-predicate holds. We therefore abstract to collections of idle states between step executions.

To prove \Lname Turing complete, we have to prove, as a base case, that after the initialization we find ourselves into a collection of idle states with the initial configuration of $T$ with $Z$ as its implementation configuration. We also have to prove as a step case that every transition taken by either $T$ with $Z$ or $\impl$ can be exactly mirrored by $\impl$ or $T$ with $Z$ s.t. after the transition, the implementation configuration is again equal to the Turing machine configuration of $T$ with $Z$. To do this, we also have to prove that after any transition taken, there is only a single non-null \textit{Control} instance in the idle states reached. From this, it follows by induction that $\impl$ and $T$ with $Z$ have equivalent behaviour, so $\impl$ is a correct implementation of $T$ with $Z$ and \Lname is Turing complete.

To prove that after initialization, the collection of idle states reached have the initial Turing machine configuration of $T$ with $Z$ as their implementation configuration, we first prove that the step \textit{init} is deterministic. We then use this to prove the exact result of executing \textit{init}, by using Corollary~\ref{cor: statef}.

\begin{lem}
\label{lem: ndinit}
The execution of \emph{init} in $\impl$ is deterministic.
\end{lem}
\begin{proof}
The step \emph{init} is only executed once, at the start of the program, by the null-instance of \textit{Control} (as no other instances exist). As only one instance exists, there cannot be a race condition between two struct instances. Therefore \textit{init} contains no race conditions.
It then follows by Lemma~\ref{lem:det} that \textit{init} is deterministic.
\end{proof}
\begin{lem}[Executing \textit{init} in the initial state]\label{contract: init}
Let $P^0_{\impl}$ be the initial state at the start of executing $\impl$ and let the input string $Z= z_0\ldots z_n$. Executing the step \textit{init} on $P^0_{\impl}$ results in a state $P_1$ with a single non-null \textit{Control} instance such that $(q_0, t_Z)$ is the implementation configuration of $P_1$.
\end{lem}
\begin{proof}
First, note that the initial state for $\impl$, as defined in the \Lname semantics, contains the schedule of $\impl$, the null-instances of all structs, and a stability stack. The stability stack has no bearing on this proof, and will be disregarded. From Lemma~\ref{lem: ndinit}, we know that as only the null-instance of \textit{Control} executes \textit{init}, the execution of \textit{init} is deterministic.

Then let $P_1$ be an idle state reached by executing \textit{init} from $P^0_{\impl}$. As written in line 43 of $\impl$, the null-instance of \textit{Control} will create one additional \textit{Control} instance. As no other struct instances execute \textit{init}, it follows that there will be a single non-null \textit{Control} instance in $P_1$. Let this instance be $c$. From line 43 of Listing~\ref{lst: tmex}, we also know that $c$ will have its state set to $0$, which is the representation of $q_0$.

We then prove that the function made according to Definition~\ref{def: implconf} in $P_1$ from the \textit{TapeCells} is $t_Z$. 
Due to Lemma~\ref{lem: ndinit}, we can walk through Listing~\ref{ex:init} sequentially. Then, through multiple applications of Corollary~\ref{cor: statef}, we know that the first part makes one \textit{TapeCell} instance for every $z_i\in Z$, and the second part then connects these to the correct \textit{left} and \textit{right} neighbours. It follows that in $P_1$, there exists a \textit{TapeCell} for all symbols $z_i\in Z$, and no others, s.t. every \textit{TapeCell} $z_i$ is be connected to $z_{i-1}$ and $z_{i+1}$ (if they exist) through parameters \textit{left} and \textit{right} respectively.
The last line of \emph{init} then sets the \textit{head} parameter of $c$ to the \textit{TapeCell} for $z_0$. It follows that the function made according to Definition~\ref{def: implconf} is indeed $t_Z$.

Then the lemma holds: the implementation configuration of $P_1$ is $(q_0, t_Z)$. \end{proof}

We then prove that taking a transition in the $\impl$ has the same effect as taking the same transition in $T$ with input string $Z$. We do this by first establishing that the step \textit{transition} taken from a state with a single non-null \textit{Control} instance is deterministic, after which we prove that all executions of \textit{transition} taken in $\impl$ are deterministic (as there will always be only a single \textit{Control} instance). We then prove that executing some transition from a configuration $(q, t)$ in $\impl$ has the same resulting configuration as executing the same transition from $(q, t)$ in $T$ with $Z$. As a part of this, we also prove that no new \textit{Control} instances are made, as that would mean that the idle states reached would not have an implementation configuration.

\begin{lem}
\label{lem: ndtransition}
Let $P$ be an idle state reachable in $\impl$ with a single non-null \textit{Control} instance.
Any execution of \emph{transition} executed from $P$ is deterministic.
\end{lem}
\begin{proof}
As per Lemma~\ref{lem:det}, it suffices to prove that the execution contains no race conditions.
Due to the definition of \textit{transition} (Listing~\ref{ex:clause}), only one clause of the step will be executed. Let $c$ be this clause.
As there is only one non-null \textit{Control} instance, let $c_0$ be the null-instance of \textit{Control} and let $c_1$ be the non-null instance of \textit{Control}.

During the execution of \emph{transition} by the null-instance of \textit{Control} $c_0$, all of the updates executed by $c_0$ will be updates to parameters of null-instances (recall that $c_0.\mathit{head} = \nil$).
As there is only one non-null \textit{Control} instance $c_1$, any race condition must involve both $c_0$ and $c_1$, and must then be a race for a null-parameter. However, race conditions require at least one write to be performed to the variable and writes to null-parameters are not performed (as per the semantic rule \textbf{ComWrNSkip}). It follows that the execution contains no race conditions.
\end{proof}
\begin{lem}
\label{lem: dettrans}
Every \emph{transition} step executed in $\impl$ is deterministic.	
\end{lem}
\begin{proof}
We prove by induction that every execution of \textit{transition} in $\impl$ starts from a state with a single non-null \textit{Control} instance.
As a base case, the first execution of \textit{transition} starts from $P_1$ as defined in Lemma~\ref{contract: init}, which has a single non-null \textit{Control} instance.

Then, let $P_i$ be a state with a single non-null \textit{Control} instance from which \textit{transition} is executed for the $i$th time. Due to Lemma~\ref{lem: ndtransition}, we know that the execution of \textit{transition} is deterministic, so we can consider the sequential execution of \textit{transition}. The step \textit{transition} is made up of multiple mutually exclusive clauses by definition. Let $C$ be the clause executed, w.l.o.g., for this execution of \textit{transition}. 
By the definition of $C$, in none of the statements in $C$ a \textit{Control} instance is created, as seen in Listing~\ref{ex:clause}. 
It follows that the state $P_i'$ resulting from the execution of \textit{transition} will also have only a single non-null \textit{Control} instance.
By Corollary~\ref{cor:prog}, we know that eventually, from $P_i'$, a state $P_{i+1}$ is reached s.t. either the schedule is empty or \textit{transition} is executed again. If the schedule is empty in $P_{i+1}$, every execution of the step \textit{transition} in $\impl$ will have been executed from a state with a single non-null \textit{Control} instance. Else, note that the Schedule Transitions as defined in the semantics of $\impl$ (Figure~\ref{fig: semSystem}) do not change the struct environment, so $P_{i+1}$ will again have a single non-null \textit{Control} instance.

It follows that every step \textit{transition} executed in $\impl$ will be executed from a state with a single non-null \textit{Control} instance. By Lemma~\ref{lem: ndtransition}, the lemma follows.
\end{proof}
\begin{lem}[Effect of a \textit{transition} execution]
\label{contract: transition}
Let $P$ be an idle state of $\impl$ from which \emph{transition} can be executed and let $(q,t)$ be the implementation configuration of $P$. Let $(q, t)$ also be the current Turing machine configuration of $T$ with input string $Z$.
Then the result of a transition in $T$ is a configuration $(q', t')$ iff the result of executing the \emph{transition} step from $P$ in $\impl$ is an idle state $P'$ such that $(q',t')$ is its implementation configuration.
\end{lem}
\begin{proof}
Let $P$ and $(q, t)$ be as defined in the lemma. Then by definition of $\impl$, there exists a single transition that can be executed from $P$ for $\impl$ iff $\delta(q, t)$ is defined. Additionally, $\delta$ is a partial function, so if $\delta(q, t)$ is defined for $(q, t)$ it is always uniquely defined for $(q, t)$. 

Let $\delta(q, t) = (q', s', D)$. Let $D = R$ (the proof of $D = L$ is analogous). Then as a transition in $T$ with $Z$ is deterministic, by Definition~\ref{def: conf}, the resulting Turing machine configuration of taking a transition from $(q, t)$ is $(q', t')$, with $$t'(i) = \begin{cases}
	z'&\text{if } i = -1\\
	t(i+1) &\text{otherwise}
\end{cases}.$$

Taking the transition from $P$ for $\impl$ is also deterministic (Lemma~\ref{lem: dettrans}), and therefore we can walk through the statements of the clause to determine its effect. By definition of $\impl$, this clause is based on the template shown in Listing~\ref{ex:clause}. Let $c$ be the single \textit{Control} instance that exists in $P$, and let $h$ be the \textit{TapeCell} instance which is referenced in the \textit{head} parameter of $c$. As per multiple applications of Corollary~\ref{cor: statef}, the result of the transition is that the \textit{state} of $c$ is updated to $q'$, the \textit{symbol} of $h$ is updated to $z'$, the \textit{accepting} parameter of $c$ is updated to whether $q'\in F$ and that the \textit{head} parameter of $c$ shifts one \textit{TapeCell} to the right (making a new \textit{TapeCell} if required).

Creating a function of the \textit{TapeCells} as in Definition~\ref{def: implconf} then results in function $t''$, s.t. $t''(-1) = z'$ and $t''(i) = t(i+1)$ for all $i \neq -1$. It follows that $t'' = t'$. Then, by Definition~\ref{def: implconf}, the implementation configuration of the resulting state is $(q', t')$.
\end{proof}		
\begin{thm}
\label{thm: tc}
\Lname is Turing complete.
\end{thm}
\begin{proof}
	We prove this by induction, with as base case that after initialization, both the $\impl$ and the turing machine $T$ with input string $Z$ have configuration $(q_0, t_Z)$. For $T$ with $Z$, this follows by definition, and for $\impl$, this follows from Lemma~\ref{contract: init}. 
	Then, as the step case, we prove that when $T$ with $Z$ and $\impl$ both start from configuration $(q, t)$ and take the same transition, both end up in configuration $(q', t')$. This follows from Lemma~\ref{contract: transition}. 
	We conclude that $T$ with $Z$ and $\impl$ have the same behaviour, and therefore $\impl$ is a correct implementation of $T$ with $Z$ in \Lname. Therefore, \Lname is Turing complete.
\end{proof}

\section{Adaptations}\label{sec: adapt}
In this section, we discuss several potential adaptations for the \Lname semantics which increase practical expressivity. First, we discuss the adaptation of the schedule to allow for more diverse looping structures, and then we will discuss the addition of array-like structures to \Lname. This will include additions to the syntax and semantics of \Lname.
Here, we consider the adaptations in this section separately from each other, and that for all adaptations we extend the original semantics of \Lname, as discussed in Section~\ref{sec:AuD}.
At the end of the section, we include a short discussion on these additions.

\subsection{Schedule Adaptations}
In \Lname, the schedule controls the flow of the program, using steps, barriers and fixpoints, as discussed in Section~\ref{sec:AuD}. 

The fixpoint has two main components as a looping structure. First, it requires synchronisation between all struct instances in each iteration of the loop. Second, it requires that the whole system is stable for the loop to terminate. 
For the schedule adaptations, we will first relax the second requirement, introducing a possible addition of \emph{parameter-specific fixpoints}, which only require certain parameters to be stable. Then, we will relax the first requirement and give a possible addition of \emph{iterators}, which removes the synchronisation requirement for fixpoints.

\subsubsection{Parameter-specific fixpoints.} In \Lname, fixpoints require that after the execution of the fixpoint every parameters in the entire system is stable. This is convenient theoretically, as this is a strong guarantee for the contents of a state after fixpoint termination. In practice, however, it could be useful to limit the relevant parameters for the fixpoint. 

To illustrate this, we have created a program based on the reachability example given earlier. It is displayed in Listing~\ref{lst: parex}. 
The goal of this program is not only to compute which nodes are reachable, but also to count the iterations necessary to do so. 
The program does this by counting the iterations of the fixpoint in a parameter $\mathit{count}$.

This program, however, will never become stable, as \textit{count} is a parameter. Furthermore, as the value of \textit{count} stays relevant between iterations of the fixpoint, \textit{count} must be a parameter.
The only way to mitigate that this program never becomes stable, is to make the program more complicated by computing stability of the graph during the step reachability and somehow making that available to all edges, which would be undesirable.
To that end, it may be useful to limit the relevant parameters for the fixpoint to exclude the parameter \textit{count}, which is not stable by design.
\begin{lstlisting}[float=t, caption={Counting Reachability.}, label={lst: parex}]
struct Node (reach: Bool) {}

struct Edge (in: Node, out: Node, diam: Int) {
	reachability {
		count := count + 1;
		if (in.reach = true) then { 
			out.reach := true;
		}
	}
	init {
		Node node1 := Node(true);
		Node node2 := Node(false);
		Node node3 := Node(false);
		Node node4 := Node(false);
		Node node5 := Node(false);
		
		Edge edge12 := Edge(node1, node2, -1);
		Edge edge13 := Edge(node1, node3, -1);
		Edge edge23 := Edge(node2, node3, -1);
		Edge edge34 := Edge(node3, node4, -1);
		Edge edge51 := Edge(node5, node1, -1);
	} 
}

init < Fix(reachability)\end{lstlisting}

To limit the relevant parameters for a fixpoint, we add \emph{parameter-specific fixpoints}: special fixpoints which allow us to specify the parameters the fixpoint requires to be stable. Note that we still consider fixpoints unstable when new struct instances are made during the execution of an iteration of the fixpoint. 
For this, we expand the syntax of schedules to allow fixpoints of the form: ``$\mathit{Fix}(\mathit{sc}: x_1, x_2, \ldots)$'', where $\mathit{sc}$ is the subschedule to be executed and $x_1, x_2, \ldots$ are the variables which to be stable for the fixpoint to terminate.
Note that this could be straightforwardly extended in the case the variables need to be uniquely identified with their struct types.
In our example in Listing~\ref{lst: parex}, we replace the fixpoint with ``$\mathit{Fix}(\mathit{reachability}: \mathit{reach})$'', which makes the program work as intended.

To allow the parameter-specific fixpoints in the syntax of \Lname, we have to add an additional entry $\word{Fix} \word{(} \type{Sched}\word{:} \type{Vars} \word{)}$ to the syntax for the schedule, which represents a parameter specific fixpoint. Here, $\type{Vars}$ gives the parameters relevant for this fixpoint, with $\type{Vars} ::= \type{Var}\type{AVars}$ and $\type{AVars} ::= \word{,} \type{Var}\type{AVars} \mid \vempty$. 

Let $\Program$ be a program for which the semantics is defined. 
We then change the definition of the semantics as follows to be able to use parameter-specific fixpoints.
First, we adapt the state to include a \emph{stability function}, which is a total function which keeps track of which fixpoints depend on which variables. 
This function has type $\mathbb{N}\rarr \Par_\Program\rarr\Bool$ and is initialized as $\false$ for all input pairs $(i, x)\in \mathbb{N}\times \mathit{Par}_\Program$. This function is called during a write to check whether this write is relevant for the stability of the fixpoints. To avoid unnecessary complication, we assume that every parameter name is unique; to remove this assumption, make the stability function type $\Nat\to(\mathit{Type}\times\mathit{Par}_\Program)\to\Bool$ instead.
The function $\mathit{sf}$ can be seen as a stack of functions, where $\mathit{sf}(k)$ denotes which parameters are relevant for the $k$th fixpoint currently initialized.
The new state is then defined as a four-tuple $\langle \Sched, \Structs, \Stab, \mathit{sf}\rangle$, with $\Sched$ the schedule, $\Structs$ the struct type environment, $\Stab$ the stability stack and $\mathit{sf}:\Nat \rarr \Par_\Program \rarr \Bool$ the \emph{stability function}. 

For our initial state, we define the function $\mathit{sf}^0$, which returns false for all input pairs $(i, x)\in \Nat\times\Par_\Program$, which is added as the fourth element of the initial state.

For the transition rules, we define the function $\mathit{sf}^0_t: \Par_\Program\rarr \Bool$, which returns true for all inputs $x\in\Par_\Program$. 		
We change the rule \textbf{ComWr} to the following rule \textbf{ComWrN}, which ensures that the stability stack is only reset when the relevant variables are changed. The changed parts of the rule are highlighted.
\[
\inference[(\textbf{ComWrN})]{
	\Structs(\ell) = \langle \StructType, \writev{x};\CList, \Stack;v;\ell',
	\Env\rangle \\
	\Structs(\ell') = \langle \StructType', \ComList', \Stack', \Env'\rangle \\
	\ell' \notin \Labels^{0} \lor x\notin \Par(\StructType', \Program) \\
	\mathit{su} = (x\notin\Par(\StructType', \Program)\lor\Env'(x) =v) \\
}{
	\begin{aligned}
		\langle \Sched, \Structs, \Stab , &\highlight{\mathit{sf}}\rangle\Rightarrow\langle \Sched, \Structs[\ell \mapsto \langle \StructType, \CList, \Stack, \Env\rangle][\ell', 4\mapsto \Env'[x\mapsto v]],\\
		&\Stab_1 \land \highlight{(\mathit{su} \lor \neg\mathit{sf}(1)(x))};\dots;\Stab_{|\Stab|} \land \highlight{(\mathit{su}\lor \neg\mathit{sf}(|\Stab|)(x))}, \highlight{\mathit{sf}}\rangle
	\end{aligned}
}
\]
In \textbf{ComWrN}, as in \textbf{ComWr}, $\mathit{su}$ is true when $x$ is not a parameter changed by the application of this rule. However, due to the addition of parameter-specific fixpoints, even if $x$ is a changed parameter, whether the entries of the stability stack	need to be reset to $\false$ also depends on whether $x$ is relevant for the entries. Let $\Stab_k$ be one of the entries of the stability stack. Then whether $x$ is relevant for the $k$th currently initialized fixpoint is saved in $\mathit{sf}(k)(x)$: if $x$ is relevant, $\mathit{sf}(k)(x) = \true$. Then $\Stab_k$ should only be reset to false (if it isn't false already) if $x$ is a relevant, changed parameter, which is reflected in updating $\Stab_k$ to $\Stab_k \land (\mathit{su} \lor\mathit{sf}(k)(x))$.

We also adapt the rule \textbf{FixInit} into the rules \textbf{FixInitG} (for the general case) and \textbf{FixInitS} (for parameter-specific fixpoints), which properly initiates the function $\mathit{sf}$ for every fixpoint. Let $X\subseteq \Par_\Program$ be some list of parameters. The differences between the new rules and \textbf{FixInit} are highlighted.
\[
\inference[(\textbf{FixInitG})]{
	\mathit{Done}(\sigma)\\
	\highlight{k = |\Stab|+1}\\
}{
	\langle \mathit{Fix}(F);F_1, \Structs, \Stab, \highlight{\mathit{sf}} \rangle\Rightarrow \langle F;\mathit{aFix}(F);F_1, \Structs,\Stab;\true, \highlight{\mathit{sf}[k\mapsto \mathit{sf}^0_t}]\rangle
}
\]
\[
\inference[(\textbf{FixInitS})]{
	\mathit{Done}(\sigma) &
	\highlight{k = |\Stab|+1}\\
	\highlight{X = x_1, \ldots, x_n} &
	\highlight{x_1, \ldots, x_n \in \Par_\Program}\\
	\highlight{\mathit{sf}' = \mathit{sf}[x\mapsto \mathit{sf}^0_f[\{x_1\mapsto \true, \ldots, x_n\mapsto \true\}]]}
}{
	\langle \mathit{Fix}(F\highlight{: X});F_1, \Structs, \Stab, \highlight{\mathit{sf}} \rangle\Rightarrow \langle F;\mathit{aFix}(F);F_1, \Structs,\Stab;\true, \highlight{\mathit{sf}'}\rangle
}
\]
In these rules, we correctly initialize normal fixpoints to fixpoints dependent on every parameter (as shown by setting $\mathit{sf}(k)$ to $\mathit{sf}^0_t$) and parameter specific fixpoints to fixpoints only dependent on the given parameters and no others (as shown by the definition of $\mathit{sf}'$ and overwriting $\mathit{sf}(k)$ with $\mathit{sf}'$.
We adapt the rules \textbf{FixIter} and \textbf{FixTerm} into the rules \textbf{FixIterN} and \textbf{FixTermN} as follows. Differences are again highlighted.
\[
\inference[(\textbf{FixIterN})]{\mathit{Done}(\sigma)}{
		\langle \mathit{aFix}(\mathit{sc})<\mathit{sc}_1, \Structs,  \Stab;\false, \highlight{\mathit{sf}} \rangle\sr\qquad\langle \mathit{sc}<\mathit{aFix}(\mathit{sc})<\mathit{sc}_1, \Structs, \Stab;\true, \highlight{\mathit{sf}}\rangle
}
\]
\[
\inference[(\textbf{FixTermN})]{\mathit{Done}(\sigma)}{
	\langle \mathit{aFix}(\mathit{sc})<\mathit{sc}_1, \Structs, \Stab;\true, \highlight{\mathit{sf}} \rangle\sr\langle \mathit{sc}_1, \Structs, \Stab, \highlight{\mathit{sf}}\rangle
}
\]
These rules are not significantly different from \textbf{FixIter} and \textbf{FixTerm}, as both resetting and initializing levels of $\mathit{sf}$ happens simultaneously in \textbf{FixInitG} and \textbf{FixInitS}.
For the command rules, we then also add $\mathit{sf}$ to both the state before and after the transition, much like in \textbf{FixIterN} and \textbf{FixTermN}.

\subsubsection{Iterators.} In \Lname, fixpoints require that between iterations, there is a synchronisation for all struct instances, through the use of the $\mathit{Done}$-predicate in the transition rules. This is due to the fact that the rules \textbf{FixIter} and \textbf{FixTerm} require the $\mathit{Done}$-predicate to be true for the state it is taken from in the semantics of \Lname (see also Section~\ref{sec:AuD}).
This makes the program more understandable, as it provides a strong guarantee for the status of variables between the executions of the fixpoint. 

However, this also inhibits performance, as synchronisation leads to overhead and there are programs for which a strong guarantee on the status of variables between executions is not necessary for understandability.
For example, consider again Listing~\ref{lst: reach} for computing reachability.
While we could synchronise between the different iterations of the reachability step, this does not lead to a better understanding, less write-write race conditions or other benefits.
Though a straightforward extension may be an alternative barrier which does not require synchronisation, we believe that this will remove much of the structure and simplicity of \Lname. 
Therefore, we instead introduce an unsynchronised loop: the \emph{iterator}, where the lack of synchronisation is contained within the iterator only.

The iterator takes the form $\mathit{Iter}(F_1;F_2;F_3;\ldots)$, with every occurrence $F_i$ a step name. For this extension, we do not allow iterators or fixpoint nested into iterators (though an iterator nested into a fixpoint is allowed).
During execution, the iterator will loop over these steps. With $\mathcal{S}_i$ the statements of the step $F_i$, every iteration, all struct instances execute the commands generated by $\interp{\mathcal{S}_1;\mathcal{S}_2;\mathcal{S}_3;\ldots}$, if they have any.
If any struct instances which had commands are finished executing their commands, the iterator can check whether the system is stable. 
If it is not, then the iterator will make sure that for all struct instances that have to execute commands, there is at least one more iteration of the commands at the end of the command list.
The iterator terminates when there are no more commands to execute and the system is stable.
This results in a loop where struct instances asynchronously execute the sequence of steps included in the iterator until the system is stable.

For example, if we were to execute the step $\mathit{reachability}$ with an iterator, we would put ``$\mathit{Iter}(\mathit{reachability})$'' in the schedule. 
This would cause the struct instances to asynchronously execute multiple iterations of the step $\mathit{reachability}$, until the system is stable.
Here, as the struct instances of $\mathit{Node}$ do not execute any commands for the step $\mathit{reachability}$, they are not taken into account for the iterator.
If we had another step $\mathit{step}$, we could then also call ``$\mathit{Iter}(\mathit{reachability};\mathit{step})$''.
This would cause the struct instances to asynchronously execute multiple iterations of the statements of $\mathit{reachability}$ followed by the statements of $\mathit{step}$ (if the struct instance has commands for these steps), until the system is stable.

To implement our iterator, we extend \Lname's syntax by adding the option $\word{Iter}\linebreak\word{(}\type{StepList}\word{)}$ for the schedule, where $\type{StepList} ::= \type{Id} \mid \type{Id};\type{StepList}$, where the $\type{Id}$ occurrences are type checked to be step names.
This step list holds the steps to be executed during the iterator, which are executed without synchronisation in between.

In the remainder of this section, let $\Program$ be some program.
We need to add transition rules to be able to implement the iterator, which are based on the transition rules for fixpoints.
First, we add a rule to initialize iterators.
This rule, \textbf{IterInit}, initialises the iterator and also introduces the meta-schedule element \emph{aIter}, which denotes an iterator which has been started but has not finished yet.
We use $F^{+}$ to denote a list of one or more steps.
We define the function $I(\ell, F^{+})$ with $\ell\in\Labels$ and $F^{+} = F_1;\ldots;F_n$, which returns the list of commands $\interp{\StatList_{\StructType_\ell}^{F_1}};\ldots;\interp{\StatList_{\StructType_\ell}^{F_n}}$.
Recall from Section~\ref{sec:AuD} that $\mathcal{C}$ is the set of all commands.
Changes from the rule \textbf{FixInit} are highlighted.

\[
\inference[(\textbf{IterInit})]{\mathit{Done}(\sigma)
}{
	\begin{aligned}
		& \langle \highlight{\mathit{Iter}(F^{+})} < \mathit{sc}_1, \Structs, \Stab \rangle\sr \langle \highlight{\mathit{aIter}(F^{+})} < \mathit{sc}_1,\\
		& \qquad\qquad\Structs[\{\ell \mapsto \langle \StructType_\ell, \highlight{I(\ell, F^{+})}, \vempty, \Env_\ell, \true\rangle \mid \Structs(\ell) = \langle \StructType_\ell, \ComList_\ell, \Stack_\ell, \Env_\ell\}],\Stab;\true\rangle
	\end{aligned}
}
\]
In the second rule, we define when an iterator needs to do another iteration. This rule works as follows:
When a struct instance is done executing all steps of the iterator, and the stability stack is false, then either some other struct instance caused this instability or the current struct instance did. As we cannot discount either option, it follows that all struct instances should execute the steps from the iterator at least one more time before the iterator can be considered stable. Therefore, the rule tops up the commands of all struct instance s.t. their command lists end in exactly one full list of commands as generated by the steps in the iterator. The rule also sets the stability stack to $\true$. This is safe: every struct instance will execute the steps at least one more time, so the execution of the iterator cannot stop before at least one full iteration has happened in which no parameter was changed and no struct instances were created. This means that the iterator successfully computes a fixed point when it terminates with this iteration transition rule. Note that in this rule, $\ell'$ is universally quantified.

\[
\inference[(\textbf{IterIter})]{
	\highlight{\exists \ell\in\Labels.(\Structs(\ell) = \langle \StructType, \vempty, \Stack, \Env\rangle \land
	I(\ell, F^{+}) \neq \vempty)}
}{
	\begin{aligned}
		& \langle \mathit{aIter}(F^{+})< \mathit{sc}_1, \Structs, \Stab;\false \rangle\sr \langle \mathit{aIter}(F^{+})< \mathit{sc}_1,\\
		& \qquad\qquad\Structs[\{\ell' \mapsto \langle \StructType_{\ell'}, \ComList_{\ell'};I(\ell', F^{+}), \Stack_{\ell'}, \Env_{\ell'}\rangle \mid\\
		&\qquad\qquad \Structs(\ell') = \langle \StructType_{\ell'}, \ComList_{\ell'}, \Stack_{\ell'}, \Env_{\ell'}\rangle \land \nexists\ComList'\in\mathcal{C}^{*}.(\ComList_{\ell'} = \ComList';I(\ell', F^{+}))\}],\Stab;\true\rangle
	\end{aligned}
}
\]
In the third rule, which can be taken when the iterator is stable and every struct instance is done, the iterator terminates. 
\[
\inference[(\textbf{IterTerm})]{
	\mathit{Done}(\sigma)\\
}{
	\begin{aligned}
		& \langle \highlight{\mathit{aIter}(F^{+})}<\mathit{sc}_1, \Structs, \Stab;\true \rangle\sr \langle \mathit{sc}_1,\Structs,\Stab\rangle
	\end{aligned}
}
\]

With the first rule, we satisfy that the iterator requires a synchronisation before starting its execution. In the second rule, we allow struct instances to start new iterations of the iterator step without synchronisation, and with the third rule, we require a synchronisation before the loop terminates. This satisfies the first requirement.
The commands put on the command lists of struct instances in the first and second rules are generated from the entire subschedule inside the iterator. As commands are run asynchronously, it follows that the second requirement also holds.
The third requirement holds because \Lname syntax does not permit fixpoints or iterators nested in iterators when expanded with the syntax additions given above.
It follows that these extensions correctly implement the iterator as intended into the semantics of \Lname.

We can use the iterator to create better code for \Lname reachability, by changing Listing~\ref{lst: reach}'s schedule to the schedule given below:
\begin{lstlisting}[float=h, caption={New Schedule for Listing~\ref{lst: reach}}, label={lst: reach2}]
	init < Iter(reachability)		\end{lstlisting}

\subsection{Arrays}
Arrays offer a quick way of grouping elements which allows for quick access due to insights in how things are stored in the memory. \Lname does not naturally support dynamically-sized arrays. This is due to the fact that while \Lname does not assume a memory implementation, it does require that all struct instances only have a constant-sized amount of memory to their disposal, as \Lname focuses on small data elements. 

As arrays are widely used in other programming languages and algorithms, supporting arrays may make it easier for people used to those other languages and algorithms to adopt \Lname.
Furthermore, it may make it easier to translate existing implementations to \Lname.
It would allow open up the possibility to leverage the constant time access arrays are known for in \Lname.
Therefore, we propose one way to extend \Lname with support for arrays. 
This would allow people to use arrays with \Lname, and would allow people to leverage the constant time access inherent to them as well.

To adopt arrays in \Lname, we introduce a new, special \emph{array instance}. The idea is that this array instance owns an array and saves the array starting address and size.
As with conventional arrays, while the size may be chosen dynamically upon initialization, it may not be changed after the array has been initialized.
The array is then filled with labels of struct instances, functioning as links to those struct instances. Each element can be accessed through indirect access. Note that this keeps the assumption of constant-sized memory intact for non-array instances, which we prefer over just building in arrays by abolishing the assumption entirely. In this extension, we also consider every entry of an array to be a parameter, to keep the extension simple.

First, we add the new expression $\word{array}(\type{Exp})$ to the syntax of \Lname. This expression reflects the creation of a new array. The expression contained in it represents the size of the array. When type checking, the inner expression should resolve to a natural number larger than $0$, and the whole program should again be well-typed.

When accessing variables, we allow for variable accesses of the type $a[E]$, where $a[E]$ is an array access with $E$ an expression. Accesses can be chained with other variable accesses when necessary. This is reflected by changing the syntax for variables to \begin{align*}
	\type{Var} &::= \type{Var}.\type{Id}\mid\type{Var}\word{[}\type{Exp}\word{]}\mid\type{Id}
	\end{align*} 
We add the type checking requirement that the expression in the square brackets has to resolve to a natural number and that if there is an array access $a[E_1]\cdots[E_n]$ in the syntax, $a$ must be of type $\mathit{Array}_1(\cdots(\mathit{Array}_n(T)))$ for some $T$. We also require updates of arrays to be well-typed. Lastly, we consider ``\texttt{array}'' to be a protected word in the syntax.

We also allow for a new expression $\type{Var}\word{.s}$, which returns the size of the array given for $\type{Var}$. During type checking, we require $\type{Var}$ to refer to an array.

Note that in this extension, we have not included array initializations like ``Int$[2]$ $a$ $:=$ $\{0, 1\}$", as we would like to keep our commands as simple as possible.
We are therefore not inclined to introduce some kind of multiwrite command. This means that the commands generated from a statement like the one given above would not be significantly different from the statements ``Array(Int) $a$ $:=$ array$(2)$; $a[0] \mathop{:=} 0$; $a[1] \mathop{:=} 1$;''. 

Let $\Program$ be an arbitrary program.
We create a new set $\mathit{Types}^a_\Program = \StructTypes_\Program \cup \{\Nat, \Int, \Bool, \String\} \cup \mathit{Types}^1_\Program \cup \mathit{Types}^2_\Program\cup\ldots$, where $$\mathit{Types}^1_\Program = \{\mathit{Array}(\StructType)\mid \StructType\in\StructTypes_\Program\} \cup \{\mathit{Array}(\Nat), \mathit{Array}(\Int), \mathit{Array}(\Bool), \mathit{Array}(\String)\}$$ and $\mathit{Types}^i_\Program = \{ \mathit{Array}(T)\mid T\in\mathit{Types}^{i-1}_\Program\}$ in the semantics. We also define that in the semantics, arrays are labelled using the set $\Labels^a\subset\Labels$.

Arrays are dependent on the memory. To that end, and to keep the struct environment separate from hardware dependent memory, we add a \emph{memory function} $\mathcal{M}:\mathcal{A}\rightharpoonup\mathcal{V}$ to the semantic state, with $\mathcal{A}$ a set of memory addresses and $\mathcal{V}$ the set of all values (with $\mathcal{A}\subset\Nat$). This function models the hardware memory, where the addresses for which it is defined are allocated addresses. We also define that our set $\mathbb{N}$ must include the number $0$. We consider the empty address to be an abstract, non-existent address $\alpha_0$.

We define an \emph{array instance} $a$ as a tuple $\langle \alpha, s\rangle$, which contains the address $\alpha$ at which $a$ starts and the size $s$ of $a$. We save array instances in the struct environment (and extend the struct environment to facilitate this). Lastly, we extend the $\nil$-instances with a $\nil$-instance $\langle \alpha_0, 0\rangle$, with label $\ell^a_0$. We extend $\mathit{defaultVal}$ such that $\mathit{defaultVal}(\texttt{Array}(T)) = \ell^a_0$ for any $T\in \StructTypes^a_\Program$.

To work with arrays, we create the commands $\readA$ and $\writeA$ with $x\in\mathit{Var}_\Program\cup\Par_\Program$ to read from and write to a location of an array. We also add the new commands $\arr(T)$, with $T\in\StructTypes^a_\Program$, and $\asize$ which takes a variable $x\in\mathit{Var}_\Program\cup\mathit{Par}_\Program$ and retrieves the size of the array saved under $x$. 

For the interpretation function, we remove the interpretation function clauses $\interp{x_1.\cdots.x_n}$, $\interp{T\ x := E}$ and $\interp{x_1,\cdots, x_n.x := E}$ and we add the following clauses for variables, with expressions $E$ and $E'$, sequences of variables $X$, $X_1$ and $X_2$ and single variable $x$. Here, $X$ is a variable expression as constructed by the syntax above. We use the notation $X[E]$ to denote that the last $\type{Var}$ element of $X$ is an array access element ``$[E]$'' and we use the notation $X.x$ to denote that the last $\type{Var}$ element of $X$ is a variable access element ``$.x$''. 
\[
\begin{split}			
	\interp{x} &= \push{\this};\readg{x}\\
	\interp{X[E]} &= \interp{X};\interp{E};\readA\\
	\interp{X.x} &= \interp{X};\readg{x}\\
	\interp{x := E'} &= \interp{E'};\push{\this};\writev{x}\\
	\interp{X[E] := E'} &= \interp{E'};\interp{X};\interp{E};\writeA\\
	\interp{X.x := E'} &= \interp{E'};\interp{X};\writev{x}\\
	\interp{T\ X := E} &= \interp{X := E}
\end{split}
\]
We also add the following clauses to deal with array creation and requesting the size of an array:
\[
\begin{split}
	\interp{\word{array}(E)} &= \interp{E};\arr(T)\\
	\interp{X.\texttt{s}} &= \interp{X};\asize\\	
\end{split}
\]

We then add new operational semantics rules to reflect reading from an array (\textbf{ComRdA}), writing to an array (\textbf{ComWrA}), getting the size of the array (\textbf{ComArrs}) and creating an array (\textbf{ComArr}). We also adapt the other operational semantics to include $\mathcal{M}$ into the state.
The rule \textbf{ComRdA} given below is based on the rule \textbf{ComRd}, with changes to access the memory instead of variable environments and to use the command \textbf{RdA}. Note that the location is realized as a displacement from the memory address of the array. Changes from \textbf{ComRd} are highlighted.

\[
\inference[(\textbf{ComRdA})]{
	\Structs(\ell) = \langle \StructType, \highlight{\readA;\CList, \Stack;\ell';v}, \Env\rangle\\
	\highlight{\Structs(\ell') = \langle \alpha, s\rangle} & v < s & \mathcal{M}(\alpha + v) \text{ is defined} 
}{
	\langle \Sched, \Structs, \highlight{\mathcal{M},} \Stab \rangle\sr\langle \Sched, \Structs[\ell \mapsto \langle \StructType, \CList, \Stack;\highlight{\mathcal{M}(\alpha + v),} \Env\rangle], \highlight{\mathcal{M},} \Stab\rangle
}
\]
The rule \textbf{ComWrA} given below is based on \textbf{ComWr}. Like \textbf{ComWr}, it keeps track of whether it has to update the stability stack if the array is changed and the array is a parameter, though this means that in the end, three entries of the struct environment are integral to the functioning of this command.
\[
\inference[(\textbf{ComWrA})]{
	\Structs(\ell) = \langle \StructType, \highlight{\writeA;\CList, \Stack;v_1;\ell';v_2},
	\Env\rangle &
	\highlight{\Structs(\ell') = \langle \alpha, s\rangle} \\
	\highlight{\ell'} \notin \Labels^{0} & v_2 < s &
	\mathit{su} = (\highlight{\mathcal{M}(\alpha + v_2) = v_1})\\
}{
	\begin{aligned}
		\langle \Sched, \Structs, \highlight{\mathcal{M},} \Stab \rangle\Rightarrow \langle \Sched, \Structs&[\ell \mapsto \langle \StructType, \CList, \Stack, \Env\rangle],\\
		&\highlight{\mathcal{M}[(\alpha+v_2) \mapsto v_1],} \Stab_1 \land \mathit{su};\dots;\Stab_{|\Stab|} \land \mathit{su}\rangle
	\end{aligned}
}
\]
The rule \textbf{ComAsize} given below is also based on \textbf{ComRd}, and returns the size of an array. The differences with \textbf{ComRd} are highlighted.
\[
\inference[(\textbf{ComAsize})]{
	\Structs(\ell) = \langle \StructType, \highlight{\asize;\CList, \Stack;\ell'}, \Env\rangle\\
	\highlight{\Structs(\ell') = \langle \alpha, s\rangle}
}{
	\langle \Sched, \Structs, \highlight{\mathcal{M},} \Stab \rangle\Rightarrow\langle \Sched, \Structs[\ell \mapsto \langle \StructType, \CList, \Stack;\highlight{s}, \Env\rangle], \highlight{\mathcal{M},}\Stab\rangle
}
\]
The last new rule is \textbf{ComArr}, which creates a new array. This rule is based on \textbf{ComCons}, and differences with that rule are highlighted. Note that $\ell'$ is a fresh label in this rule, signified by $\Structs(\ell') = \bot$.
\[
\inference[(\textbf{ComArr})]{
	\Structs(\ell) = \langle \StructType, \highlight{\arr(T);\CList, \Stack;s}, \Env\rangle\\
	\highlight{s\in\mathbb{N}} & \highlight{s\geq 1} & \highlight{\text{$\mathcal{M}$ is not defined for $\alpha, \ldots, (\alpha+s)$}} &
	\Structs(\ell') = \bot
}{
	\begin{aligned}
		\langle \Sched, \Structs, \highlight{\mathcal{M},} \Stab \rangle\Rightarrow&\langle \Sched, \Structs[\{\ell \mapsto \langle \StructType, \CList, \Stack;\ell', \Env\rangle, \highlight{\ell'\mapsto \langle \alpha, s\rangle}\}],\\
		& \highlight{\mathcal{M}[\{\alpha_i\mapsto \mathit{DefaultVal}(T) \mid \alpha_i\in \alpha, \ldots, (\alpha+s)\}],\Stab\rangle}
	\end{aligned}
}
\]

With the rule \textbf{ComArr}, we can now make array instances, allocating them a sequence of unallocated addresses. We can also get the size of the array through the addition \texttt{.s} behind an array variable \texttt{x}, which is then handled by the transition \textbf{ComAsize}, induced by the \textbf{asize} command. We can lastly read and write from those arrays using conventional block-parenthesis notation, and the locations we read and write to can be determined dynamically. This satisfies our expectations for an array. Note that this extension does not include the destruction and deallocation of arrays.

With this extension, we are able to reimplement reachability to use arrays, as shown in Listing~\ref{lst: reacharr} (based on Listing~\ref{lst: reach}). As you can see in the Listing, arrays make the system smaller, but as loops are contained in the schedule, one cannot employ a for-loop to walk through all array instances in a single step, and the code must be adapted for this.

\begin{lstlisting}[float=t, caption={\Lname code for a reachability program on a small graph using arrays}, label={lst: reacharr}]
struct Node (reach: Bool, succ: Array(Node), done: Nat) {
	reachability {
		if (reach = true) then { 
			if (done < succ.s){
				succ[done].reach := true;
				done := done + 1;
			}
		}
	}
	
	init {
		Node node1 := Node(true, array(2), 0);
		Node node2 := Node(false, array(1), 0);
		Node node3 := Node(false, array(1), 0);
		Node node4 := Node(false, null, 0);
		
		node1.succ[0] := node2;
		node1.succ[1] := node3;
		node2.succ[0] := node3;
		node3.succ[0] := node4;
	} 
}
init < Fix(reachability)\end{lstlisting}

\subsection{Discussion}\label{sec: discussion}
In this section, we presented three extensions for \Lname to allow for more streamlined and conventional programming in \Lname. These extensions are reasonably simple and therefore follow \Lname's goals of being a small programming language. Of these extensions, we think that the extension for parameter-specific fixpoints has the most merit, as it simplifies \Lname code and allows for more concise programs, which plays into \Lname's goals. We also estimate that it should be relatively easy to implement this extension and not have too much overhead in the execution time of \Lname programs, as part of the stability function can probably be derived at compile time.

The iterator increases the amount of parallelism in the schedule and therefore follows the goals of \Lname to have as much parallelism as possible, but for this extension in particular it would be important to test the extension to find out the impact it has on the speed of execution of \Lname iterators compared to fixpoints before adding it.
If it does not have too much of a negative impact on execution speed, this would be a good extension for \Lname.

The array extension is directly and theoretically useful to convert older parallel algorithms to \Lname and because many programmers will already be used to them. 
However, the impact on the efficiency of \Lname will have to be tested. 
Additionally, using arrays may lead to less parallel code: traversing an array is a sequential process, which can lead to a delayed application of new information to elements in the array. For example, where in Listing~\ref{lst: reach} the edges propagated reachability information obtained the next step without extra delay, the first Node in Listing~\ref{lst: reacharr} will only try to set Node 3 to reachable in the second step. 
As separating steps from loops is a fundamental design choice of \Lname, we furthermore do not think it is a good idea to solve this problem by introducing limited loops in the steps.
However, even if it were solved that way, such a simple array traversal loop would make step execution less uniform and more sequential, as the struct instance with the longest array walked through now has to finish before the next step can be executed, and it could be that this array is significantly larger than the other arrays.	
Lastly, our current array extension considers all elements of any array to be parameters, which is a problem for fixpoints if a step uses a local array. 
To solve this, one will most likely have to start counting per array which parameters it is the value of, leading to more complicated extensions and read and write procedures.
We therefore estimate that to make an array-like structure work in \Lname, more extensive extensions outside of the scope of this paper would be required, and would require a consideration of how arrays should be defined beyond the standard to work best in \Lname.

We have kept ourselves to extensions that increase the practical expressivity in some way. 
Two important extensions that this paper did not cover are an extension for \Lname to do garbage collection and also remove struct instances, and an extension to include inheritance in \Lname and to be able to import and create packages and templates of structs in some way.
As both of these extensions are significant and will most likely impact the basis of \Lname (for example, garbage collection and destroying struct instances will most likely impact the corollaries used in the Turing completeness proof), they will require a careful approach to implement them and we therefore deemed them outside of the scope of this paper. 
A similarly significant extension that does impact expressivity is to have an asynchronous separator which does separate steps but does not require synchronisation in between, which will at the very least impact the structured approach of \Lname and requires careful design outside of the scope of this paper. These three extensions are thus left for future work.

For this paper, we have not implemented and tested these extensions. We consider this the limitation of the currently expressed extensions, as the practical usability of \Lname is important, despite it not being our foremost concern. Additionally, to formalize the exact effects of the extensions on \Lname executions, we could develop a formal proof system for \Lname, extend it for each of the extensions and prove it sound.

\section{Conclusion}\label{sec: Con}
In this paper, we have proven \Lname Turing complete by giving a method to implement any Turing machine in \Lname. Additionally, we have explored the practical expressivity of \Lname by seeking out the limits of \Lname's syntax and presented extensions to \Lname to overcome these limits.
With this paper and the papers by Leemrijse~\etal~\cite{leemrijse-formalisation-2025, leemrijse2023}, we have made a case that \Lname is practically feasible, general-purpose and reasonably adaptable to specific purposes. 
This sets the basis for \Lname to be generally applicable.

However, more work can be done to streamline \Lname's use enough to make it a more attractive option for applications. First of all, we should make some theoretical applications that capitalize on \Lname's theoretical foundation, like a proof system for \Lname. On the practical side, there are at least two larger extensions that \Lname could benefit from, garbage collection and package management, which could be incorporated into \Lname. Additionally, we can continue our search for useful \Lname extensions and optimize the practical expressivity of \Lname more. 

\section*{Acknowledgment}
\noindent The authors would like to thank Gijs Leemrijse for his work on implementing algorithms in AuDaLa.
\bibliographystyle{alphaurl}
\bibliography{ExpressivityAuDaLa}

\newcommand{\etalchar}[1]{$^{#1}$}
\begin{thebibliography}{HND{\etalchar{+}}21}

\bibitem[BY95]{baba--netl-1995}
T.~Baba and T.~Yoshinaga.
\newblock A-{NETL}: a language for massively parallel object-oriented
  computing.
\newblock In {\em PMMPC Proc.}, pages 98--105. IEEE, 1995.
\newblock \href {https://doi.org/10.1109/PMMPC.1995.504346}
  {\path{doi:10.1109/PMMPC.1995.504346}}.

\bibitem[CDK14]{chong-sound-2014}
Nathan Chong, Alastair~F. Donaldson, and Jeroen Ketema.
\newblock {A Sound and Complete Abstraction for Reasoning about Parallel Prefix
  Sums}.
\newblock {\em SIGPLAN Not.}, 49(1):397–409, 2014.
\newblock \href {https://doi.org/10.1145/2578855.2535882}
  {\path{doi:10.1145/2578855.2535882}}.

\bibitem[Cop24]{copeland-church-turing-1997}
B.~Jack Copeland.
\newblock {The Church-Turing Thesis}.
\newblock In Edward~N. Zalta and Uri Nodelman, editors, {\em The {Stanford}
  Encyclopedia of Philosophy}. Metaphysics Research Lab, Stanford University,
  {W}inter 2024 edition, 2024.

\bibitem[dB{\etalchar{+}}12]{de-boer-decidability-2012}
Frank~S. de~Boer et~al.
\newblock Decidability {Problems} for {Actor} {Systems}.
\newblock In {\em {CONCUR} 2012 – {Concurrency} {Theory}}, volume~10 of {\em
  Logical Methods in Computer Science}, pages 562--577. Springer, 2012.
\newblock \href {https://doi.org/10.1007/978-3-642-32940-1_39}
  {\path{doi:10.1007/978-3-642-32940-1_39}}.

\bibitem[DD02]{detrey-constructive-2002}
J{\'e}r{\'e}mie Detrey and Oliver Diessel.
\newblock {\em {A Constructive Proof of the Turing Completeness of Circal}}.
\newblock School of Computer Science and Engineering, University of New South
  Wales, Australia, 2002.

\bibitem[DG09]{di2009expressiveness}
Cinzia Di~Giusto.
\newblock {\em Expressiveness of Concurrent Languages}.
\newblock PhD thesis, Alma Mater Studiorum Università di Bologna, 2009.

\bibitem[DK98]{deursen-little-1998}
Arie~Van Deursen and Paul Klint.
\newblock Little languages: little maintenance?
\newblock {\em Journal of Software Maintenance: Research and Practice},
  10:75--92, 1998.
\newblock \href
  {https://doi.org/10.1002/(SICI)1096-908X(199803/04)10:2<75::AID-SMR168>3.0.CO;2-5}
  {\path{doi:10.1002/(SICI)1096-908X(199803/04)10:2<75::AID-SMR168>3.0.CO;2-5}}.

\bibitem[FN24]{franken-audala-2024}
Tom T.~P. Franken and Thomas Neele.
\newblock {AuDaLa} is {Turing} {Complete}.
\newblock In {\em FORTE 2024 Proc.}, volume 14678 of {\em LNCS}, pages
  221--229. Springer Nature Switzerland, 2024.
\newblock \href {https://doi.org/10.1007/978-3-031-62645-6_12}
  {\path{doi:10.1007/978-3-031-62645-6_12}}.

\bibitem[FNG23]{franken-autonomous-2023}
Tom T.~P. Franken, Thomas Neele, and Jan~Friso Groote.
\newblock An {Autonomous} {Data} {Language}.
\newblock In {\em Theoretical {Aspects} of {Computing} – {ICTAC} 2023},
  volume 14446 of {\em {LNCS}}, pages 158--177. Springer International
  Publishing, 2023.

\bibitem[FNG25]{franken-autonomous-2025}
Tom T.~P. Franken, Thomas Neele, and Jan~Friso Groote.
\newblock The {Autonomous} {Data} {Language} – {Concepts}, design and formal
  verification.
\newblock {\em Theoretical Computer Science}, 1057:115560, 2025.
\newblock URL:
  \url{https://www.sciencedirect.com/science/article/pii/S0304397525004980},
  \href {https://doi.org/10.1016/j.tcs.2025.115560}
  {\path{doi:10.1016/j.tcs.2025.115560}}.

\bibitem[G{\etalchar{+}}08]{garland-parallel-2008}
Michael Garland et~al.
\newblock Parallel {Computing} {Experiences} with {CUDA}.
\newblock {\em IEEE Micro}, 28(4):13--27, 2008.
\newblock \href {https://doi.org/10.1109/MM.2008.57}
  {\path{doi:10.1109/MM.2008.57}}.

\bibitem[Gib15]{gibbons-functional-2015}
Jeremy Gibbons.
\newblock Functional {Programming} for {Domain}-{Specific} {Languages}.
\newblock In {\em {CEFP} 2013}, {LNCS}, pages 1--28. Springer International
  Publishing, 2015.
\newblock \href {https://doi.org/10.1007/978-3-319-15940-9_1}
  {\path{doi:10.1007/978-3-319-15940-9_1}}.

\bibitem[HMU01]{hopcroft-introduction-2001}
John~E. Hopcroft, Rajeev Motwani, and Jeffrey~D. Ullman.
\newblock {\em Introduction to automata theory, languages, and computation}.
\newblock Addison-Wesley, Boston, 2nd edition, 2001.

\bibitem[HND{\etalchar{+}}21]{henderson-turing-2021}
Alec Henderson, Radu Nicolescu, Michael~J. Dinneen, T.~N. Chan, Hendrik Happe,
  and Thomas Hinze.
\newblock Turing completeness of water computing.
\newblock {\em J Membr Comput}, 3:182--193, 2021.
\newblock \href {https://doi.org/10.1007/s41965-021-00081-3}
  {\path{doi:10.1007/s41965-021-00081-3}}.

\bibitem[Koz76]{kozen-parallelism-1976}
Dexter Kozen.
\newblock {On parallelism in Turing machines}.
\newblock In {\em 17th {Annual} {Symposium} on {Foundations} of {Computer}
  {Science} (sfcs 1976)}, pages 89--97. IEEE, 1976.
\newblock \href {https://doi.org/10.1109/SFCS.1976.20}
  {\path{doi:10.1109/SFCS.1976.20}}.

\bibitem[Lee23]{leemrijse2023}
Gijs Leemrijse.
\newblock {Towards relaxed memory semantics for the Autonomous Data Language},
  2023.
\newblock {MSc.} thesis, Eindhoven University of Technology.

\bibitem[LFN25]{leemrijse-formalisation-2025}
Gijs~P. Leemrijse, Tom T.~P. Franken, and Thomas Neele.
\newblock Formalisation of a {New} {Weak} {Semantics} for {AuDaLa}.
\newblock In {\em Automated {Technology} for {Verification} and {Analysis}},
  pages 93--116. Springer Nature Switzerland, 2025.
\newblock \href {https://doi.org/10.1007/978-3-031-78750-8_5}
  {\path{doi:10.1007/978-3-031-78750-8_5}}.

\bibitem[QYZG17]{qu-parallel-2017}
Peng Qu, Jin Yan, You-Hui Zhang, and Guang~R. Gao.
\newblock Parallel {Turing} {Machine}, a {Proposal}.
\newblock {\em J. Comput. Sci. Technol.}, 32:269--285, 2017.
\newblock \href {https://doi.org/10.1007/s11390-017-1721-3}
  {\path{doi:10.1007/s11390-017-1721-3}}.

\bibitem[RK{\etalchar{+}}17]{ragan-kelley-halide-2017}
Jonathan Ragan-Kelley et~al.
\newblock Halide: decoupling algorithms from schedules for high-performance
  image processing.
\newblock {\em Commun. ACM}, 61:106--115, 2017.
\newblock \href {https://doi.org/10.1145/3150211} {\path{doi:10.1145/3150211}}.

\bibitem[RL93]{raimbault-relacs-1993}
F.~Raimbault and D.~Lavenier.
\newblock {RELACS} for systolic programming.
\newblock In {\em ASAP Proc.}, pages 132--135. IEEE, 1993.
\newblock \href {https://doi.org/10.1109/ASAP.1993.397128}
  {\path{doi:10.1109/ASAP.1993.397128}}.

\bibitem[UA10]{ungar-harnessing-2010}
David Ungar and Sam~S. Adams.
\newblock Harnessing emergence for manycore programming: early experience
  integrating ensembles, adverbs, and object-based inheritance.
\newblock In {\em OOPSLA Proc.}, pages 19--26. ACM, 2010.
\newblock \href {https://doi.org/10.1145/1869542.1869546}
  {\path{doi:10.1145/1869542.1869546}}.

\bibitem[Wie84]{wiedermann-parallel-1984}
Juraj Wiedermann.
\newblock {\em Parallel Turing machines}.
\newblock Department of Computer Science, University of Utrecht The
  Netherlands, 1984.

\bibitem[Y{\etalchar{+}}17]{yamashita-turing-2017}
Tatsuya Yamashita et~al.
\newblock Turing-{Completeness} of {Asynchronous} {Non}-camouflage {Cellular}
  {Automata}.
\newblock In {\em Cellular {Automata} and {Discrete} {Complex} {Systems}},
  {LNCS}, pages 187--199. Springer International Publishing, 2017.
\newblock \href {https://doi.org/10.1007/978-3-319-58631-1_15}
  {\path{doi:10.1007/978-3-319-58631-1_15}}.

\end{thebibliography}
\end{document}